\documentclass[twoside,final,11pt]{entics}
\pdfoutput=1
\usepackage{enticsmacro}
\usepackage{cite}
\usepackage{amsmath,amssymb,amsfonts}
\usepackage{algorithmic}
\usepackage{graphicx}
\usepackage{enumitem}
\usepackage{textcomp}
\usepackage{etoc}
\usepackage{xcolor}
\usepackage{cleveref}
\usepackage{amssymb}
\usepackage{amsmath}

\AddToHook{env/lem/begin}{\crefalias{thm}{lemma}}
\AddToHook{env/defn/begin}{\crefalias{thm}{definition}}
\AddToHook{env/rem/begin}{\crefalias{thm}{remark}}
\AddToHook{env/prop/begin}{\crefalias{thm}{proposition}}
\AddToHook{env/exmp/begin}{\crefalias{thm}{example}}
\crefname{equation}{}{}
\crefname{equations}{}{}
\crefname{figure}{Fig.}{Fig.}
\crefname{definition}{Definition}{Definitions}
\crefname{example}{Example}{Example}
\crefname{remark}{Remark}{Remarks}
\crefname{thm}{Theorem}{Theorems}
\crefname{proposition}{Proposition}{Propositions}
\crefname{lemma}{Lemma}{Lemmas}
\crefname{section}{Section}{Sections}
\crefname{app}{Appendix}{Appendices}

\usepackage{todonotes}
\usepackage{stmaryrd}
\usepackage{quiver}
\usepackage{bussproofs}
\newenvironment{bprooftree}
{\leavevmode\hbox\bgroup}
{\DisplayProof\egroup}

\usepackage{tikz}
\usetikzlibrary{arrows,decorations,automata,backgrounds,positioning,calc,cd}
\usetikzlibrary{shapes.geometric,shapes.misc,matrix,arrows.meta,decorations.markings}
\usetikzlibrary{circuits.logic.US,patterns,shapes}
\usepackage{mathcommand,mathtools}
\usepackage{xfrac}
\newlength{\ketketspacing}
\newlength{\nketketspacing}
\newcommand{\ket}[1]{\mathord{\lvert #1 \rangle}\futurelet\next\ketspacing}
\newcommand{\nket}[1]{\mathord{\lvert #1 ]}\futurelet\next\nketspacing}
\def\ketspacing{
    \ifx\next\ket
        \hspace{\ketketspacing}
    \else\ifx\next\nket
        \hspace{\ketketspacing}
    \else
        \relax
    \fi\fi
}
\def\nketspacing{
    \ifx\next\ket
        \hspace{\nketketspacing}
    \else\ifx\next\nket
        \hspace{\nketketspacing}
    \else
        \relax
    \fi\fi
}
\usepackage{arydshln}

\makeatletter
\def\stringify#1{\expandafter\zap@space\detokenize{#1} \@empty}
\makeatother
\usepackage{tikzit}
\ifdef{\macrotikzfig}{}{\newcommand{\macrotikzfig}[2]{\input{#1.tikz}}}
\tikzstyle{black}=[fill=black, draw=black, shape=circle, minimum size=6pt, inner sep=0pt, outer sep=0pt]
\tikzstyle{dist}=[fill=none, draw=black, isosceles triangle, isosceles triangle apex angle=60, text height=0.6em, inner sep=0pt, shape border rotate=180, anchor=center, shape border uses incircle]
\tikzstyle{white}=[fill=white, draw=black, shape=circle, minimum size=6pt, inner sep=0pt, outer sep=0pt, tikzit draw=black, tikzit fill=white]
\tikzstyle{basic box}=[draw, inner sep=2pt, fill=white, rectangle, minimum height=1.2em, minimum width=1em]
\tikzstyle{and}=[fill=white, draw=black, and gate US, scale=0.6]
\tikzstyle{or}=[fill=white, draw=black, or gate US, anchor=center, scale=0.6]
\tikzstyle{not}=[fill=white, draw=black, not gate US, anchor=center, scale=0.6]
\tikzstyle{if}=[trapezium, trapezium angle=70, draw, inner xsep=0pt, outer sep=0pt, minimum height=15pt, rotate=270, text width=8pt, fill=white]
\tikzstyle{myif}=[trapezium, trapezium angle=60, draw, inner xsep=0pt, outer sep=0pt, minimum height=15pt, rotate=270, text width=10pt, fill=white, label={center:\rotatebox{270}{$\mathtt{if}$}}]
\tikzstyle{myif2}=[trapezium, trapezium angle=60, draw, inner xsep=0pt, outer sep=0pt, minimum height=10pt, rotate=270, text width=7pt, fill=white, label={center:\rotatebox{270}{$\mathtt{if}$}}]
\tikzstyle{convex}=[trapezium, trapezium angle=70, draw, inner xsep=0pt, outer sep=0pt, minimum height=15pt, rotate=270, text width=7pt, fill=white, label={center:$+_{\!p}$}]
\tikzstyle{convexp1}=[trapezium, trapezium angle=70, draw, inner xsep=0pt, outer sep=0pt, minimum height=15pt, rotate=270, text width=7pt, fill=white, label={center:$+_{\!\!p_{\b1}}$}]
\tikzstyle{convexp0}=[trapezium, trapezium angle=70, draw, inner xsep=0pt, outer sep=0pt, minimum height=15pt, rotate=270, text width=7pt, fill=white, label={center:$+_{\!\!p_{\b0}}$}]
\tikzstyle{convexq}=[trapezium, trapezium angle=70, draw, inner xsep=0pt, outer sep=0pt, minimum height=15pt, rotate=270, text width=7pt, fill=white, label={center:$+_{\!q}$}]
\tikzstyle{flip}=[draw={rgb,255: red,86; green,86; blue,86}, fill={rgb,255: red,86; green,86; blue,86}, rounded rectangle, rounded rectangle right arc=none, text=white, node font={\scriptsize}, inner sep=0.2em]
\tikzstyle{flipop}=[draw={rgb,255: red,86; green,86; blue,86}, fill={rgb,255: red,86; green,86; blue,86}, rounded rectangle, rounded rectangle left arc=none, text=white, node font={\scriptsize}, inner sep=0.2em]
\tikzstyle{flipup}=[draw={rgb,255: red,204; green,112; blue,0}, fill={rgb,255: red,204; green,112; blue,0}, rounded rectangle, rounded rectangle right arc=none, rotate=90, text=black, node font={\scriptsize}, inner sep=0.2em]
\tikzstyle{single}=[-, draw=black]
\tikzstyle{double}=[-, draw=black, line width=2pt, tikzit fill=cyan, tikzit draw=cyan]
\tikzstyle{boxes}=[-, fill=white, dotted]
\tikzstyle{unit}=[-, fill=white, dotted]

\newcommand{\id}{\mathrm{id}}
\newcommand{\tensor}{\otimes} 

\newcommand{\FStoch}{\mathsf{FStoch}}
\newcommand{\BStoch}{\mathsf{BStoch}}

\newcommand{\FSprod}{\FStoch^\otimes}
\newcommand{\BSprod}{\BStoch^\otimes}
\newcommand{\FSsum}{\FStoch^\oplus}
\newcommand{\FSklprod}{\FStoch_{\scriptscriptstyle \mathsf{kl}}^\otimes}

\newcommand{\BSklprod}{\BStoch_{\scriptscriptstyle \mathsf{kl}}^\otimes}
\newcommand{\BSreprod}[1]{\BStoch_{\scriptscriptstyle \textsf{ré}#1}^\otimes}
\newcommand{\FSklsum}{\FStoch_{\scriptscriptstyle \mathsf{kl}}^\oplus}
\newcommand{\FSresum}[1]{\FStoch_{\scriptscriptstyle \textsf{ré}#1}^\oplus}

\newcommand{\N}{\mathbb{N}}

\newcommand{\B}{\mathbb{B}}

\newrobustcmd{\Hclosesum}[1]{#1^{\mathsf{H}\vten}}
\newrobustcmd{\Pclosesum}[1]{#1^{\mathsf{P}\vten}}
\newrobustcmd{\Hclosemax}[1]{#1^{\mathsf{H}\vmeet}}
\newrobustcmd{\Pclosemax}[1]{#1^{\mathsf{P}\vmeet}}
\newrobustcmd{\Hclosestar}[1]{#1^{\mathsf{H}\star}}
\newrobustcmd{\Pclosestar}[1]{#1^{\mathsf{P}\star}}
\newrobustcmd{\tensorvten}{\mathbin{\boxtimes_{\vten}}}
\newrobustcmd{\tensorvmeet}{\mathbin{\boxtimes_{\vmeet}}}
\newrobustcmd{\tensorplus}{\mathbin{\boxtimes_{+}}}
\newrobustcmd{\tensormax}{\mathbin{\boxtimes_{\max}}}

\newrobustcmd{\Hclosesumsum}[1]{#1^{\mathsf{H}\vten\vten}}
\newrobustcmd{\Pclosesumsum}[1]{#1^{\mathsf{P}\vten\vten}}
\newrobustcmd{\Hclosesummax}[1]{#1^{\mathsf{H}\vten\vmeet}}
\newrobustcmd{\Pclosesummax}[1]{#1^{\mathsf{P}\vten\vmeet}}
\newrobustcmd{\Hclosemaxsum}[1]{#1^{\mathsf{H}\vmeet\vten}}
\newrobustcmd{\Pclosemaxsum}[1]{#1^{\mathsf{P}\vmeet\vten}}
\newrobustcmd{\Hclosemaxmax}[1]{#1^{\mathsf{H}\vmeet\vmeet}}
\newrobustcmd{\Pclosemaxmax}[1]{#1^{\mathsf{P}\vmeet\vmeet}}

\newrobustcmd{\anyclose}[1]{\overline{#1}}
\newcommand{\vle}{\sqsubseteq}
\newcommand{\vge}{\sqsupseteq}
\newcommand{\vjoin}{\sqcup}
\newcommand{\vJoin}{\bigsqcup}
\newcommand{\vMeet}{\bigsqcap}
\newcommand{\vten}{+}
\newcommand{\vmeet}{\sqcap}
\newcommand{\vun}{k}

\newcommand{\zeroinfQ}{{[0,\infty]_{\scriptscriptstyle +}}}

\newrobustcmd{\preordqmat}[1]{\mathsf{POHA}_{#1}}
\newrobustcmd{\metqmat}[1]{\mathsf{MHA}_{#1}}
\newrobustcmd{\RuleRefl}{\hyperlink{rulerefl}{\textnormal{\textsc{Refl}}}}
\newrobustcmd{\RuleTriang}{\hyperlink{ruletriang}{\textnormal{\textsc{Triang}}}}
\newrobustcmd{\RuleBot}{\hyperlink{rulebot}{\textnormal{\textsc{Bot}}}}
\newrobustcmd{\RuleMon}{\hyperlink{rulemon}{\textnormal{\textsc{Mon}}}}
\newrobustcmd{\RuleJoin}{\hyperlink{rulejoin}{\textnormal{\textsc{Join}}}}
\newrobustcmd{\RuleSeqsum}{\hyperlink{ruleseqsum}{$\textnormal{\textsc{Seq}}_{\vten}$}}
\newrobustcmd{\RuleTenssum}{\hyperlink{ruletenssum}{$\textnormal{\textsc{Par}}_{\vten}$}}
\newrobustcmd{\RuleSeqmeet}{\hyperlink{ruleseqmeet}{$\textnormal{\textsc{Seq}}_{\vmeet}$}}
\newrobustcmd{\RuleTensmeet}{\hyperlink{ruletensmeet}{$\textnormal{\textsc{Par}}_{\vmeet}$}}
\newrobustcmd{\RuleCont}{\hyperlink{rulecont}{$\textnormal{\textsc{Cont}}$}}
\newrobustcmd{\RuleSym}{\hyperlink{rulesym}{$\textnormal{\textsc{Symm}}$}}
\newrobustcmd{\letter}{f}
\newrobustcmd{\type}{n}

\newrobustcmd{\cc}[1]{\renewrobustcmd{\letter}{#1}\macrotikzfig{ccletter}{#1}}
\newrobustcmd{\distdiag}[1]{\renewrobustcmd{\letter}{#1}\macrotikzfig{distletter}{#1}}
\newrobustcmd{\statediag}[2]{\renewrobustcmd{\letter}{#1}\renewrobustcmd{\type}{#2}\macrotikzfig{stateletter}{#1}}
\newrobustcmd{\flip}[1]{\renewrobustcmd{\letter}{#1}\macrotikzfig{flip}{#1}}
\newmathcommand{\distvect}[2]{\scalebox{0.55}{$\begin{bmatrix} #1\\#2 \end{bmatrix}$}}
\makeatletter
\DeclareFontEncoding{LS1}{}{}
\DeclareFontEncoding{LS2}{}{\noaccents@}
\DeclareFontSubstitution{LS1}{stix}{m}{n}
\DeclareFontSubstitution{LS2}{stix}{m}{n}
\makeatother
\DeclareMathAlphabet\mathscr{LS1}{stixscr}{m}{n}
\SetMathAlphabet\mathscr{bold}{LS1}{stixscr}{b}{n}

\newrobustcmd{\Dset}{\mathcal{D}}
\newrobustcmd{\Dmet}{\overline{\mathcal{D}}}
\newrobustcmd{\dist}{\varphi}
\newrobustcmd{\distb}{\psi}
\newrobustcmd{\distc}{\tau}
\newrobustcmd{\Dist}{\Phi}
\newrobustcmd{\Distb}{\Psi}
\newrobustcmd{\dirac}[1]{\ket{#1}}
\newrobustcmd{\supp}[1]{\mathrm{supp}(#1)}
\newrobustcmd{\tv}{\mathsf{tv}}
\newrobustcmd{\tvmax}{\mathsf{tv}_{\times}}
\newrobustcmd{\kl}{\mathsf{kl}}
\newrobustcmd{\klmax}{\overline{\mathsf{kl}}}
\newrobustcmd{\tvplus}{\mathsf{tv}_{\otimes}}
\newrobustcmd{\re}[1]{\textsf{ré}_{#1}}
\newrobustcmd{\remax}[1]{\overline{\textsf{ré}}_{#1}}
\newrobustcmd{\Kant}[1]{#1_{\mathrm{K}}}
\newrobustcmd{\Cpl}{\mathfrak{C}}
\newrobustcmd{\discrete}{d_{\top}}
\newrobustcmd{\fset}[1]{\underline{#1}}
\newrobustcmd{\drel}[1]{\Delta_{#1}}
\newrobustcmd{\ndrel}[1]{\Delta^{\!\mathsf{c}}_{#1}}
\newrobustcmd{\Circ}{\scriptscriptstyle\mathsf{CC}}
\newrobustcmd{\scomp}{\mathbin{;}}
\newrobustcmd{\elist}{\mathsf{e}}

\newrobustcmd{\qenrcats}[1]{{#1}\mathsf{HMet}}

\newrobustcmd{\ConvAlg}{\scriptscriptstyle\mathsf{CA}}
\newrobustcmd{\HA}[1]{\mathsf{HA}_{#1}}
\newrobustcmd{\BarAlg}{\mathsf{Lib}}
\newrobustcmd{\KLsum}{\scriptscriptstyle \mathsf{KL}\oplus}
\newrobustcmd{\KLprod}{\scriptscriptstyle \mathsf{KL}\otimes}
\newrobustcmd{\Rprod}[1]{\scriptscriptstyle \mathsf{R}_{#1}^{\otimes}}
\newrobustcmd{\Rsum}[1]{\scriptscriptstyle \mathsf{R}_{#1}^{\oplus}}
\newrobustcmd{\Chainsum}{\textsc{Chain}_{\oplus}}
\newrobustcmd{\Chainprod}{\textsc{Chain}_{\otimes}}
\newrobustcmd{\Ifmax}{\textsc{If}_{\max}}
\newrobustcmd{\Parmax}{\textsc{Par}_{\max}}

\newrobustcmd{\lists}[1]{#1^*}

\newrobustcmd{\eqn}{\varphi}
\newrobustcmd{\eqnb}{\psi}
\newrobustcmd{\eqns}{\Gamma}
\newrobustcmd{\eqnsb}{\Delta}
\newrobustcmd{\assign}{\iota}

\newrobustcmd{\ctxO}{\mathsf{O}}
\newrobustcmd{\ctxM}{\mathsf{M}}
\newrobustcmd{\Ctx}{\mathbf{Ctx}}
\newrobustcmd{\sat}{\vDash}
\newrobustcmd{\dom}{\mathsf{dom}}
\newrobustcmd{\cod}{\mathsf{cod}}
\newrobustcmd{\sem}[1]{\left\llbracket #1 \right\rrbracket}
\renewrobustcmd{\th}{\mathscr{T}}

\newrobustcmd{\del}{\raisebox{0.3ex}{\scalebox{0.6}{\begin{tikzpicture}
	\begin{pgfonlayer}{nodelayer}
		\node [style=black] (5) at (-0.5, 0) {};
		\node [style=none] (6) at (0.5, 0) {};
	\end{pgfonlayer}
	\begin{pgfonlayer}{edgelayer}
		\draw (5) to (6.center);
	\end{pgfonlayer}
\end{tikzpicture}
}}}
\newrobustcmd{\cop}{\raisebox{0.3ex}{\scalebox{0.6}{\input{cop.tikz}}}}
\newrobustcmd{\ccOneMinus}{\raisebox{0.3ex}{\scalebox{0.6}{\input{ccOneMinus.tikz}}}}
\newrobustcmd{\swap}[1]{\sigma_{#1}}
\newrobustcmd{\sigCA}{\Sigma_{\mathsf{CA}}}
\newrobustcmd{\img}[1]{\vcenter{\hbox{\includegraphics[scale=0.15]{#1}}}}
\newrobustcmd{\one}{I}
\makeatletter
\newrobustcmd{\bigplus}{
	\DOTSB\mathop{\mathpalette\mattos@bigplus\relax}\slimits@
}
\newcommand\mattos@bigplus[2]{
	\vcenter{\hbox{
			\sbox\z@{$#1\sum$}
			\resizebox{!}{0.9\dimexpr\ht\z@+\dp\z@}{\raisebox{\depth}{$\m@th#1+$}}
	}}
	\vphantom{\sum}
}
\makeatother
\newrobustcmd{\sig}{\Sigma}
\newrobustcmd{\source}{\mathsf{in}}
\newrobustcmd{\target}{\mathsf{out}}
\newrobustcmd{\syncat}[2]{\mathcal{S}_{\scriptstyle #1,#2}}
\newrobustcmd{\qsyncat}[2]{\widehat{\mathcal{S}}_{\scriptstyle #1,#2}}
\newrobustcmd{\sync}[1]{\mathcal{S}_{\scriptstyle #1}}
\newrobustcmd{\vync}[1]{\mathcal{M}_{\scriptstyle #1}}
\newrobustcmd{\LT}[1]{\mathcal{L}_{\scriptstyle #1}}
\newrobustcmd{\Eq}[1]{\mathfrak{E}_{\scriptstyle #1}}
\newrobustcmd{\QEq}[1]{\mathfrak{Q}_{\scriptstyle #1}}
\newrobustcmd{\QI}[1]{\mathfrak{I}_{\scriptstyle #1}}

\LoopCommands\lettersUppercase[cat#1]
{\newmathcommand#2{\mathbf#1}}
\LoopCommands\lettersLowercase[tt#1]
{\newmathcommand#2{\mathtt#1}}
\LoopCommands\lettersLowercase[vv#1]
{\newmathcommand#2{\vec#1}}
\LoopCommands\lettersUppercase[v#1]
{\newmathcommand#2{\mathcal#1}}
\LoopCommands\lettersUppercase[vRel#1]
{\newmathcommand#2{\mathcal#1\mathbf{Rel}}}
\renewrobustcmd{\b}[1]{\mathtt{#1}}

\newrobustcmd{\Label}[1]{\RightLabel{\textnormal{\textsc{#1}}}}

\newcounter{theqn}
\def\conf{MFPS 2026} 	
\volume{6}			
\def\volu{6}			
\def\papno{16}		
\def\mydoi{18985}	
\def\lastname{Sarkis and Zanasi} 
\def\runtitle{Diagrammatic Axiomatisation of Relative Entropies} 
\def\copynames{R. Sarkis, F. Zanasi}    
		
\begin{document}
\begin{frontmatter}
	\title{Complete Diagrammatic Axiomatisations \\of Relative Entropy} 						
	\thanks[ALL]{This work was partially funded by the ARIA Safeguarded AI programme.}   
	\author{Ralph Sarkis\thanksref{a}}
	\author{Fabio Zanasi\thanksref{a}}
	\address[a]{Department of Computer Science\\ University College London\\				
		London, United Kingdom}
	\begin{abstract}
		Relative entropy is a fundamental class of distances between probability distributions, with widespread applications in probability theory, statistics, and machine learning. In this work, we study relative entropy from a categorical perspective, viewing it as a quantitative enrichment of categories of stochastic matrices. We consider two natural monoidal structures on stochastic matrices, given by the Kronecker product and the direct sum. Our main results are complete axiomatisations of Kullback--Leibler divergence and, more generally, of R\'enyi divergences of arbitrary order, for each such structure. Our axiomatic theories are formulated within the framework of quantitative monoidal algebra, using a graphical language of string diagrams enriched with quantitative equations.
	\end{abstract}
\end{frontmatter}
\section{Introduction}

Programming language semantics traditionally investigates when two programs are \emph{equivalent} with respect to a given notion of observation, or whether a program satisfies properties such as correctness, termination, or safety. When programs exhibit random behaviour, however---as in probabilistic programming, statistical inference, and machine learning---equivalence is often too coarse to be practically informative. Rather than asking whether two programs produce identical outputs, it is more useful to measure \emph{how far apart} their behaviours are. This shift in perspective has led to extensive research on behavioural distances and program metrics~(\!\!\cite{Desharnais99,vBWorrell2005,Crubille15,Baldan2018,DalLago2023}).

A prominent line of work in this direction is the programme of \emph{quantitative algebraic theories}, initiated by Mardare et al.\@~(\!\!\cite{Mardare2016}). This framework provides axiomatic presentations of (pseudo)metrics via calculi whose judgements have the form $s =_{\varepsilon} t$, expressing that the distance between terms $s$ and $t$ is at most $\varepsilon \in \mathbb{R}$. Subsequent developments have extended the framework to Markov processes~(\!\!\cite{Bacci2018,Bacci2018a}), higher-order settings~(\!\!\cite{DalLago2022}), quantitative logics~(\!\!\cite{Bacci2023,Bacci2024}), string diagrams~(\!\!\cite{Lobbia2025}), and enriched algebraic theories~(\!\!\cite{Rosicky2023,Rosicky2024}).

Complete axiomatisations for the Kantorovich metric and total variation distance, introduced in~\cite{Mardare2016}, are chief examples of quantitative algebraic theories. In this paper, we focus on another fundamental class of distances between probability distributions: relative entropy. In particular, R\'enyi divergences~(\!\!\cite{Renyi1960}), which include the Kullback--Leibler (KL) divergence~(\!\!\cite{Kullback1951}), play a central role across mathematics and computer science, with applications in statistical inference and machine learning~(\!\!\cite{Fox2012,KingmaW13,Goodfellow2020}), differential privacy~(\!\!\cite{Dwork2014,Mironov2017}), information geometry~(\!\!\cite{Nielsen2020}), and program semantics~(\!\!\cite{Barthe2013,Sato2019}). Unlike the Kantorovich metric and total variation distance, however, no complete quantitative algebraic theory has yet been developed for KL divergence or, more generally, for relative entropies. Our work fills that gap.

Following~\cite{Lobbia2025}, we formulate our results at the level of symmetric monoidal categories, presenting our axiomatic calculi in the graphical language of \emph{string diagrams} within the framework of \emph{monoidal algebra}~(\!\!\cite{Selinger_2010,PiedeleuZanasi2025}). This choice is motivated mainly by two considerations. First, it allows us to reason axiomatically about distances between \emph{multidimensional} systems, such as stochastic matrices, rather than merely probability distributions. Indeed, stochastic matrices do not form a cartesian category; instead, they assemble into a symmetric monoidal category $\FStoch$, which supports a rich analysis of Bayesian inference, causality, and probabilistic graphical models~(\!\!\cite{JacobsKZ21,Lorenz2023,Lorenzin2025}). Second, it allows us to build on a substantial body of work on diagrammatic axiomatisations of probabilistic computation (particularly~\cite{Fritz09,Piedeleu2025b}, but \emph{cf.} also~\cite{Fritz2020,Stein2025,TorresRuiz2026}) and to extend it in a quantitative direction.

We shall study two enrichments of $\FStoch$ with KL divergence, distinguished by their monoidal structure. One uses the Kronecker product of matrices, denoted $\FSprod$, and the other the direct sum, denoted $\FSsum$. Both structures are natural and useful in different contexts. The category $\FSprod$ provides the standard setting for synthetic probability theory~(\!\!\cite{Fritz2020}) and for applications to Bayesian networks and causal reasoning~(\!\!\cite{JacobsKZ21,Lorenz2023,Lorenzin2025}). In contrast, $\FSsum$ is closely related to convex sets and barycentric algebras, which have been extensively studied since at least the '50s~(\!\!\cite{Stone1949PostulatesFT}), and to the interpretation of randomness as a monadic effect (see e.g.~\cite{Bonchi22,MioV20,MioSarkisVign21}).
Our main contributions are as follows.
\begin{enumerate}[nosep]
	\item \textbf{KL divergence.} We provide axiomatisations for KL divergence enrichments of $\BSprod$ (the restriction of $\FSprod$ to objects of the form $2^n$ for $n \in \mathbb{N}$)\footnote{This restriction is due to the fact that non-enriched $\BSprod$ has been axiomatised, in~\cite{Piedeleu2025b}, but no analogous result holds for the full $\FSprod$.} and of $\FSsum$~(\cref{defn:thklprod,defn:thklsum}), and show their completeness~(\cref{thm:axbsprodkl,thm:axfssumkl}).
	      Concretely, we present quantitative string diagrammatic theories whose freely generated enriched SMCs are isomorphic to enriched $\BSprod$ and $\FSsum$, respectively. A key ingredient in our axiomatisations is the chain rule, which decomposes the divergence between joint distributions in terms of the divergences between their conditionals. We express this decomposition as a quantitative implication: the premises bound the distances between conditionals, while the conclusion bounds the distance between the corresponding joint distributions. Two such implications, $\Chainprod$ and $\Chainsum$, appear in the axiom systems for $\BSprod$ and $\FSsum$, respectively, and play a central role in the completeness proofs~(\cref{lem:prod:isometryondist,lem:isometryondist}).

	\item \textbf{R\'{e}nyi divergences.} With minor adaptations, we show that our method extends to an axiomatisation of R\'enyi divergences of arbitrary order $\alpha \in [0,\infty]$, again for both monoidal structures. KL divergence arises as the special case $\alpha = 1$.

	\item \textbf{Implicational quantitative diagrammatic reasoning.} To state our results, we need to extend the notion of quantitative monoidal theory introduced in~\cite{Lobbia2025} to allow axioms formulated as implications between quantitative equations. This mirrors the original formulation of quantitative equational reasoning in a cartesian setting~(\!\!\cite{Mardare2016}), but now considered in the monoidal setting. The resulting framework is of independent interest: implicational axioms arise, for example, in the study of uniform traces~(\!\!\cite{Hasegawa2003,Bonchi2025}) and balanced Markov categories~(\!\!\cite{DiLavore2025,Fritz2026}).
\end{enumerate}

\textbf{Related work.} As noted above, quantitative algebraic theories provide the broader context for our results. Since $(\FSsum)^{\mathrm{op}}$ is cartesian, our axiomatisation could in principle yield an analogous result within the original framework of~\cite{Mardare2016}, which is based on classical cartesian algebra (we return to this observation in the conclusions). However, this is not the case for $\BSprod$ (and $\FSprod$), which therefore requires the more general setting of monoidal algebra.

Our work builds directly on~\cite{Lobbia2025}, where $\FSsum$ enriched with total variation distance is axiomatised. A different diagrammatic axiomatisation of the same enrichment for $\FSprod$ appears in~\cite{diGiorgio2025}, albeit outside the formal framework of~\cite{Lobbia2025}. By introducing implicational axioms for diagrammatic calculi, we reconcile these approaches and bring~\cite{diGiorgio2025} within the scope of quantitative monoidal algebra.

KL divergence originates in Shannon's information theory~(\!\!\cite{Shannon1948}) and has been axiomatised multiple times using functional equations~(\!\!\cite{Hobson1969,Kannappan1973,Kannappan1974}). More recently, categorical approaches have led to novel and simpler characterisations~(\!\!\cite{Baez2014,Fullwood2021,Gagne2023}). In these works, a category of stochastic processes based on $\FStoch$ is defined, and it is shown that any functor satisfying a small set of natural postulates must compute KL divergence. Related ideas also yield a concise non-categorical axiomatisation in~\cite{Leinster2019}. Perrone~(\!\!\cite{Perrone2024}) proves that relative entropy defines an enrichment of $\FSprod$, without characterising it. Our results provide precisely such a characterisation. We also mention works on categorical characterisations of (non-relative) entropy \cite{Baez2011,Parzygnat2022}.

\textbf{Synopsis.} \cref{sec:prelim} recalls necessary background on string diagrams, quantales, enriched categories, and relative entropy. \cref{sec:enrichmonth} presents our extension of quantitative monoidal theories with implicational axioms. We define the logical system, describe how it freely generates syntactic categories, and establish its soundness and completeness with respect to models in enriched monoidal categories. \cref{sec:axkl} develops the axiomatisations of KL divergence for $\BSprod$ and $\FSsum$. \cref{sec:axren} extends these results to the full family of R\'{e}nyi divergences. We conclude with a discussion of related and future work.

\section{Background}
\label{sec:prelim}

\subsection{Monoidal Algebra}\label{sec:prelim:diagrams}
We assume familiarity with basic category theory and in particular with monoidal categories (see e.g.~\cite{MacLane71}). We recall notions in monoidal algebra, the study of algebraic structures borne in monoidal categories. A detailed account can be found in~\cite{PiedeleuZanasi2025}. In the sequel, SMC is short for symmetric strict monoidal category.

\begin{definition}\label{defn:monsig}
	A \emph{monoidal signature} is the data of a set $\sig_0$ of \emph{sorts}, a set $\sig_1$ of \emph{generators}, and two maps $\source,\target: \sig_1 \rightarrow \sig_0^*$ assigning to each generator $f$ a finite list of sorts for its \emph{inputs} and for its \emph{outputs}. We write $f:u \rightarrow v \in \sig_1$ to mean that $f$ is a generator with inputs $\source(f) = u$ and outputs $\target(f) = v$. The empty list will be denoted by $\elist$.  We often write $\sig$ to refer to the tuple $(\sig_0,\sig_1,\source,\target)$.

	Given a monoidal signature $\sig$, let $\vync{\sig}$ denote the set of all monoidal terms built out of generators in $\sig$ with the following rules:
	\begin{gather*}
		\begin{bprooftree}
			\AxiomC{$f:u \rightarrow v \in \sig_1 $}
			\Label{Gen}
			\UnaryInfC{$f: u \rightarrow v \in \vync{\sig}$}
		\end{bprooftree}
		\begin{bprooftree}
			\AxiomC{$x,y \in \sig_0$}
			\Label{Swap}
			\UnaryInfC{$\sigma_{x,y}:xy \rightarrow yx \in \vync{\sig}$}
		\end{bprooftree}
		\begin{bprooftree}
			\AxiomC{$x \in \sig_0$}
			\Label{id}
			\UnaryInfC{$\id_x: x \rightarrow x \in \vync{\sig}$}
		\end{bprooftree}
		\begin{bprooftree}
			\AxiomC{$\phantom{\sig_0}$}
			\Label{Unit}
			\UnaryInfC{$\id_{\elist}: \elist \rightarrow \elist \in \vync{\sig}$}
		\end{bprooftree}\\[0.3em]
		\begin{bprooftree}
			\AxiomC{$f:u \rightarrow v\in \vync{\sig}$}
			\AxiomC{$g: v \rightarrow w\in \vync{\sig}$}
			\Label{Seq}
			\BinaryInfC{$f\scomp g : u \rightarrow w\in \vync{\sig}$}
		\end{bprooftree}
		\begin{bprooftree}
			\AxiomC{$f:u \rightarrow v\in \vync{\sig}$}
			\AxiomC{$f:u' \rightarrow v'\in \vync{\sig}$}
			\Label{Par}
			\BinaryInfC{$f \otimes f' : uu' \rightarrow vv'\in \vync{\sig}$}
		\end{bprooftree}
	\end{gather*}
	Throughout the paper, we represent monoidal terms with string diagrams: a term $f: u \rightarrow v \in \vync{\sig}$ is drawn as $\input{fuv.tikz}$, and the rules above are depicted as follows. We often omit input and output labels.
	\[ \swap{x,y} = \raisebox{0.3ex}{\scalebox{0.7}{$\input{swapxy.tikz}$}} \qquad \id_x = \raisebox{0.3ex}{\scalebox{0.7}{$\begin{tikzpicture}
	\begin{pgfonlayer}{nodelayer}
		\node [style=none] (0) at (-0.5, 0) {};
		\node [style=none] (1) at (0.5, 0) {};
		\node [style=none] (2) at (-1, 0) {$x$};
		\node [style=none] (3) at (1, 0) {$x$};
		\node [style=none] (4) at (0, 0.5) {};
	\end{pgfonlayer}
	\begin{pgfonlayer}{edgelayer}
		\draw (0.center) to (1.center);
	\end{pgfonlayer}
\end{tikzpicture}
$}} \qquad \id_{\elist} = \raisebox{0.3ex}{\scalebox{0.7}{$\begin{tikzpicture}
	\begin{pgfonlayer}{nodelayer}
		\node [style=none] (0) at (-0.5, 0) {};
		\node [style=none] (1) at (0.5, 0) {};
		\node [style=none] (2) at (0, 0.5) {};
	\end{pgfonlayer}
	\begin{pgfonlayer}{edgelayer}
		\draw [style=unit] (0.center) to (1.center);
	\end{pgfonlayer}
\end{tikzpicture}
$}} \qquad f\scomp g = \raisebox{0.3ex}{\scalebox{0.7}{$\input{seqcompfguw.tikz}$}} \qquad f \otimes f' = \raisebox{0.3ex}{\scalebox{0.7}{$\input{parcompffprime.tikz}$}}\]
	We consider string diagrams modulo the smallest congruence (with respect to $\scomp$ and $\otimes$) containing the equations in \cref{fig:axiomssmc}. Intuitively, deformations that do not break strings or change the order of inputs and outputs do not alter the meaning of the picture. String diagrams assemble into an SMC, also denoted by $\vync{\sig}$. Its objects are finite lists of sorts in $\sig_0$, morphisms $u \rightarrow v$ are string diagrams $\input{fuv.tikz}$, composition is described by \textsc{Seq}, and monoidal product is concatenation on objects and \textsc{Par} on morphisms.
	\setcounter{theqn}{0}
	\begin{figure}[!htb]
		\begin{minipage}{0.33\textwidth}
			\begin{gather*}
				\raisebox{0.4ex}{\scalebox{0.7}{\input{monoidal/assocseqcomp.tikz}}} \\[0.4em]
				\raisebox{0.4ex}{\scalebox{0.7}{\input{monoidal/neutralid.tikz}}} \\[0.4em]
				\raisebox{0.4ex}{\scalebox{0.7}{\input{monoidal/neutralidzer.tikz}}}
			\end{gather*}
		\end{minipage}
		\begin{minipage}{0.17\textwidth}
			\begin{gather*}
				\raisebox{0.4ex}{\scalebox{0.7}{\input{monoidal/assocparcomp.tikz}}}
			\end{gather*}
		\end{minipage}
		\begin{minipage}{0.25\textwidth}
			\begin{gather*}
				\raisebox{0.4ex}{\scalebox{0.7}{\input{monoidal/interchange.tikz}}}  
			\end{gather*}
		\end{minipage}
		\begin{minipage}{0.20\textwidth}
			\begin{gather*}
				\ \ \raisebox{0.4ex}{\scalebox{0.7}{\input{monoidal/overthewire.tikz}}} \\[0.4em]
				\raisebox{0.4ex}{\scalebox{0.7}{\input{monoidal/swapswap.tikz}}} \\
			\end{gather*}
		\end{minipage}
		\caption{Laws of symmetric monoidal categories. Dotted boxes denote bracketing that is later omitted thanks to the laws.}\label{fig:axiomssmc}
	\end{figure}
\end{definition}
\begin{definition}
	A \emph{monoidal theory} $\th$ comprises a signature $\sig$ and a set $E$ of \emph{axioms}, which are pairs of monoidal terms $f,g \in \vync{\sig}$, with matching inputs and outputs, denoted $f=g$ to convey their interpretation. A \emph{model} of $\th$ is an SMC $\catC$ equipped with a symmetric monoidal functor $\sem{-}: \vync{\sig} \rightarrow \catC$, such that $\sem{f} = \sem{g}$ for all $f=g \in E$. Every such model factors uniquely through the \emph{syntactic category} $\sync{\th}$ built as a quotient of $\vync{\sig}$ by the smallest congruence containing the axioms in $E$.
\end{definition}
A recurring goal in monoidal algebra is the axiomatisation of an SMC $\catC$ with a monoidal theory, that is, to find $\th = (\sig,E)$ and a model $\vync{\sig} \rightarrow \catC$ that factors through an isomorphism $\sync{\th} \cong \catC$. We say that $\th$ \emph{presents} $\catC$. Our main results are based on existing axiomatisations of SMCs of stochastic matrices.

\begin{example}
	The category $\FStoch$ has natural numbers as objects, and as morphisms $n \rightarrow m$, the $m\times n$ stochastic matrices (matrices with entries in $[0,1]$ whose columns each sum up to $1$). When $n = 0$, there is a unique empty stochastic matrix $[]: n \rightarrow m$, and there are no morphisms $n \rightarrow 0$ when $n > 0$. Composition is by matrix multiplication and identity morphisms are matrices with $1$ on every entry of the diagonal.

	We consider two different monoidal products on stochastic matrices: the Kronecker product, denoted $\otimes$, and the direct sum, denoted $\oplus$. Both are standard notions in linear algebra and specifically for stochastic matrices. They are defined below in block matrix form with their units and symmetries.
	\[
		A \otimes B \coloneq \scalebox{0.8}{$\begin{bmatrix} a_{11}B & \dots  & a_{1n}B \\
                \vdots  & \ddots & \vdots  \\
                a_{m1}B & \dots  & a_{mn}B\end{bmatrix}$} \quad I^\otimes \coloneq 1 \quad \swap{n,m}^{\otimes} \coloneq \scalebox{0.8}{$\begin{bmatrix} E_{11} & \dots  & E_{1n} \\
                \vdots & \ddots & \vdots \\
                E_{m1} & \dots  & E_{mn}\end{bmatrix}$} \qquad
		A \oplus B \coloneq \scalebox{0.8}{$\begin{bmatrix} A & 0 \\
                0 & B\end{bmatrix}$} \quad I^\oplus \coloneq 0 \quad \swap{n,m}^{\oplus} \coloneq \scalebox{0.8}{$\begin{bmatrix} 0   & I_m \\
                I_n & 0\end{bmatrix}$}
		,\]
	where $E_{ij}$ is the $n\times m$ matrix with a single $1$ at column $i$, row $j$, and $I_n$ is the identity $n\times n$ matrix. Note that $\otimes$ acts by multiplication on objects, while $\oplus$ acts by addition.

	These form the SMCs $\FSprod$ and $\FSsum$. Only the latter has been fully axiomatised with a monoidal theory, but there is an axiomatisation of a full subcategory of $\FSprod$: $\BSprod$, the SMC of stochastic matrices of dimensions that are powers of $2$. Its objects are still natural numbers, $\BStoch(n,m) = \FStoch(2^n,2^m)$, and $\otimes$ is still the Kronecker product, but it acts by addition on objects (since $2^{n+m} = 2^n2^m$). We briefly recall the monoidal theories $\th_{\Circ}$ and $\th_{\ConvAlg}$ that present $\BSprod$ and $\FSsum$ respectively.

	The monoidal theory $\th_{\!\Circ}$ (standing for \emph{causal circuits}) has generators \raisebox{0.3ex}{\scalebox{0.7}{$\begin{tikzpicture}
	\begin{pgfonlayer}{nodelayer}
		\node [style=black] (0) at (0.5, 0) {};
		\node [style=none] (1) at (-0.5, 0) {};
	\end{pgfonlayer}
	\begin{pgfonlayer}{edgelayer}
		\draw [style=single] (0) to (1.center);
	\end{pgfonlayer}
\end{tikzpicture}
$}}, \raisebox{0.3ex}{\scalebox{0.7}{$\input{cccopy.tikz}$}}, \raisebox{0.3ex}{\scalebox{0.7}{$\input{ccand.tikz}$}}, \raisebox{0.3ex}{\scalebox{0.7}{$\begin{tikzpicture}
	\begin{pgfonlayer}{nodelayer}
		\node [style=none] (0) at (-0.8, 0) {};
		\node [style=none] (1) at (0.875, 0) {};
		\node [style=not] (2) at (0, 0) {};
	\end{pgfonlayer}
	\begin{pgfonlayer}{edgelayer}
		\draw [style=single] (0.center) to (2.input);
		\draw [style=single] (2.output) to (1.center);
	\end{pgfonlayer}
\end{tikzpicture}
$}}, and \raisebox{0.3ex}{\scalebox{0.7}{$\flip{p}$}} for each $p \in [0,1]$, and axioms as in \cref{fig:causcircaxiom} below. We write $\sync{\Circ}$ for its syntactic category.

	\begin{figure}[!htb]
		{
			\tiny
			\begin{gather*}
				\scalebox{0.8}{\input{cc/assoc.tikz}}\qquad \scalebox{0.8}{\input{cc/unit.tikz}}  \qquad \scalebox{0.8}{\input{cc/comm.tikz}} \\[0.5em]
				\scalebox{0.8}{\input{cc/and_assoc.tikz}} \qquad \scalebox{0.8}{\input{cc/and_unit.tikz}}  \qquad \scalebox{0.8}{\input{cc/and_comm.tikz}}          \\[0.5em]
				\scalebox{0.8}{\input{cc/notnot.tikz}} \qquad \scalebox{0.8}{\input{cc/copyand.tikz}} \qquad \scalebox{0.8}{\input{cc/copynotand.tikz}} \qquad \scalebox{0.8}{\input{cc/andordist.tikz}} \\[0.5em]
				\scalebox{0.8}{\input{cc/andcopy.tikz}} \qquad \scalebox{0.8}{\input{cc/notcopy.tikz}} \qquad \scalebox{0.8}{\input{cc/anddel.tikz}} \qquad \scalebox{0.8}{\input{cc/notdel.tikz}} \\[0.5em]
				\scalebox{0.8}{\input{cc/copyzero.tikz}} \qquad \scalebox{0.8}{\input{cc/copyone.tikz}} \qquad \scalebox{0.8}{\input{cc/pdel.tikz}} \qquad \scalebox{0.8}{\input{cc/pnot.tikz}}\\[0.5em]
				\scalebox{0.8}{\input{cc/geninverse.tikz}} \ \scalebox{0.8}{\input{cc/ifassoc.tikz}} \ \scalebox{0.8}{\input{cc/ifdisintegration.tikz}}
			\end{gather*}
		}
		\caption{Axiomatisation of $\BSprod$ from~\cite[Fig.~4]{Piedeleu2025b}, where $p, q, r, s, t$ are universally quantified over $[0,1]$, $\tilde{r} \coloneq rp + (1-r)q$, $\tilde{p}\coloneq \sfrac{rp}{\tilde{r}}$, $\tilde{q} \coloneq \sfrac{r(1-p)}{1-\tilde{r}}$, $\tilde{s} \coloneq st$, and $\tilde{t} \coloneq \sfrac{s(1-t)}{1-st}$. When the denominators are $0$, the values of $\tilde{p}$, $\tilde{q}$, and $\tilde{t}$ are discarded with diagrammatic reasoning, so we can define them to be anything in $[0,1]$.}
		\label{fig:causcircaxiom}
	\end{figure}
	We record below some syntactic sugar used in \cref{fig:causcircaxiom} and later. We define the or gate, the if gates, and the convex combination gates for each $p \in [0,1]$ (we use thin and thick strings together to explicitly distinguish between the type $1$ and an arbitrary type $n$):
	\[ \scalebox{0.8}{\input{ccor.tikz}} \coloneq \scalebox{0.8}{\input{ordef.tikz}} \quad \scalebox{0.65}{\input{ccif.tikz}} \coloneq \scalebox{0.8}{\input{ifdef.tikz}} \quad \scalebox{0.65}{\input{ifbunch.tikz}} \coloneq \scalebox{0.65}{\input{ifbunchdef.tikz}} \quad \scalebox{0.8}{\input{convexsugar.tikz} }\coloneq \scalebox{0.65}{\input{convexdef.tikz}}\]

	In~\cite[Theorem~3.13]{Piedeleu2025b}, it is shown that the assignment $\sem{-}: \vync{\sig_{\Circ}} \rightarrow \BSprod$ defined below yields a model of $\th_{\!\Circ}$ that factors through an isomorphism $\sync{\Circ} \cong \BSprod$.
	\begin{equation}\label{eqn:sembsprod}
		\sem{\raisebox{0.3ex}{\scalebox{0.7}{$\begin{tikzpicture}
	\begin{pgfonlayer}{nodelayer}
		\node [style=black] (0) at (0.5, 0) {};
		\node [style=none] (1) at (-0.5, 0) {};
	\end{pgfonlayer}
	\begin{pgfonlayer}{edgelayer}
		\draw [style=single] (0) to (1.center);
	\end{pgfonlayer}
\end{tikzpicture}
$}}} = [] \qquad
		\sem{\raisebox{0.3ex}{\scalebox{0.7}{$\input{cccopy.tikz}$}}} = \left[ \begin{smallmatrix}
				1 & 0 \\
				0 & 0 \\
				0 & 0 \\
				0 & 1
			\end{smallmatrix} \right] \qquad
		\sem{\raisebox{0.3ex}{\scalebox{0.7}{$\input{ccand.tikz}$}}} = \left[\begin{smallmatrix}
				1 & 0 & 0 & 0 \\
				0 & 1 & 1 & 1
			\end{smallmatrix}\right] \qquad
		\sem{\raisebox{0.3ex}{\scalebox{0.7}{$\begin{tikzpicture}
	\begin{pgfonlayer}{nodelayer}
		\node [style=none] (0) at (-0.8, 0) {};
		\node [style=none] (1) at (0.875, 0) {};
		\node [style=not] (2) at (0, 0) {};
	\end{pgfonlayer}
	\begin{pgfonlayer}{edgelayer}
		\draw [style=single] (0.center) to (2.input);
		\draw [style=single] (2.output) to (1.center);
	\end{pgfonlayer}
\end{tikzpicture}
$}}} = \left[\begin{smallmatrix}
				0 & 1 \\
				1 & 0
			\end{smallmatrix}\right]
	\end{equation}
	We note that under this correspondence, the diagram $\scalebox{0.7}{\input{convexsugar.tikz}}$ really acts by convex combination. Namely, if we precompose two diagrams $f: m \rightarrow n$ and $g: m' \rightarrow n$ whose semantics are two stochastic matrices of dimension $2^n \times 2^m$ and $2^n \times 2^{m'}$ respectively, then the resulting diagram's semantics is the matrix obtained by convex combining the columns of $\sem{f}$ with those of $\sem{g}$, with coefficient $p$. Explicitly,
	\begin{equation}\label{eqn:semconvex}
		\sem{\scalebox{0.8}{\begin{tikzpicture}
	\begin{pgfonlayer}{nodelayer}
		\node [style=basic box] (24) at (0, 0) {$f$};
		\node [style=none] (25) at (1, 0) {};
		\node [style=none] (27) at (-1, 0) {};
	\end{pgfonlayer}
	\begin{pgfonlayer}{edgelayer}
		\draw [style=double] (24) to (25.center);
		\draw [style=double] (27.center) to (24);
	\end{pgfonlayer}
\end{tikzpicture}
}} = A , \sem{\scalebox{0.8}{\begin{tikzpicture}
	\begin{pgfonlayer}{nodelayer}
		\node [style=basic box] (24) at (0, 0) {$g$};
		\node [style=none] (25) at (1, 0) {};
		\node [style=none] (27) at (-1, 0) {};
	\end{pgfonlayer}
	\begin{pgfonlayer}{edgelayer}
		\draw [style=double] (24) to (25.center);
		\draw [style=double] (27.center) to (24);
	\end{pgfonlayer}
\end{tikzpicture}
}} = B  \implies \sem{\scalebox{0.8}{\input{convfg.tikz}}} = \begin{bmatrix} A_1+_p B_1 & A_1+_pB_2 & \cdots & A_{2^m}+_p B_{2^{m'}}\end{bmatrix},
	\end{equation}
	where $A_i$ denotes the $i$th column of $A$ and similarly for $B$, and $A_i+_p B_j \coloneq pA_i+(1-p)B_j$.

	The monoidal theory $\th_{\!\ConvAlg}$ (standing for \emph{convex algebras}) has generators \raisebox{0.3ex}{\scalebox{0.7}{$\begin{tikzpicture}
	\begin{pgfonlayer}{nodelayer}
		\node [style=black] (5) at (-0.5, 0) {};
		\node [style=none] (6) at (0.5, 0) {};
	\end{pgfonlayer}
	\begin{pgfonlayer}{edgelayer}
		\draw (5) to (6.center);
	\end{pgfonlayer}
\end{tikzpicture}
$}}, \raisebox{0.3ex}{\scalebox{0.7}{$\input{cop.tikz}$}}, and \raisebox{0.2ex}{\scalebox{0.7}{$\cc{p}$}} for each $p \in [0,1]$, and axioms as in \cref{fig:fritzaxiom} below. We write $\sync{\ConvAlg}$ for its syntactic category.
	\begin{figure}[!htb]
		{
			\tiny
			\begin{gather*}
				\scalebox{0.8}{\input{assoc.tikz}} \qquad \scalebox{0.8}{\input{comm.tikz}} \qquad \scalebox{0.8}{\input{unit.tikz}}\qquad
				\scalebox{0.8}{\input{convassoc.tikz}} \qquad \scalebox{0.8}{\input{convcomm.tikz}} \\[0.5em]
				\scalebox{0.8}{\input{natdel.tikz}} \qquad \scalebox{0.8}{\input{zprob.tikz}}  \qquad \scalebox{0.8}{\input{idemp.tikz}} \qquad  \scalebox{0.8}{\input{cccop.tikz}}
			\end{gather*}
		}
		\caption{Axiomatisation of $\FSsum$ from~\cite[Definition~3.1]{Fritz09}, where $\tilde{p} = pq$ and $\tilde{q}=\frac{p -pq}{1-pq}$ (with $\frac{0}{0} = 1$).}
		\label{fig:fritzaxiom}
	\end{figure}

	In~\cite[Theorem~3.14]{Fritz09}, it is shown that the assignment $\sem{-}:\vync{\sig_{\ConvAlg}} \rightarrow \FSsum$ defined below yields a model of $\th_{\!\ConvAlg}$ that factors through an isomorphism $\sync{\ConvAlg} \cong \FSsum$.
	\begin{equation}\label{eqn:semfssum}
		\sem{\raisebox{0.3ex}{\scalebox{0.7}{$\begin{tikzpicture}
	\begin{pgfonlayer}{nodelayer}
		\node [style=black] (5) at (-0.5, 0) {};
		\node [style=none] (6) at (0.5, 0) {};
	\end{pgfonlayer}
	\begin{pgfonlayer}{edgelayer}
		\draw (5) to (6.center);
	\end{pgfonlayer}
\end{tikzpicture}
$}}} = [] \qquad
		\sem{\raisebox{0.3ex}{\scalebox{0.7}{$\input{cop.tikz}$}}} = \left[\begin{smallmatrix}
				1 & 1
			\end{smallmatrix}\right] \qquad
		\sem{\raisebox{0.3ex}{\scalebox{0.7}{$\cc{p}$}}} = \distvect{p}{1-p}
	\end{equation}
\end{example}
\subsection{Quantale-Valued Relations}
While our main applications are about axiomatising real-valued distances, our theoretical framework is presented in more general terms with quantale-valued relations.
\begin{definition}
	A \emph{quantale} $\vV$ consists of a set $\vV$ equipped with a partial order $\vle$ and a binary operation $\vten : \vV \times \vV \rightarrow \vV$ satisfying the following properties:
	\begin{itemize}[nosep]
		\item any subset $S \subseteq \vV$ has a supremum $\vJoin S$ and an infimum $\vMeet S$, so that $(\vV, \vle)$ is a \emph{complete lattice}. In particular, $\vV$ has a \emph{bottom element} $\bot = \vJoin \emptyset = \vMeet \vV$ and a \emph{top element} $\top = \vJoin \vV = \vMeet \emptyset$;
		\item $\vten$ is \emph{associative}, namely, for any $x,y,z \in \vV$, $x \vten (y \vten z) = (x \vten y) \vten z$; and
		\item $\vten$ \emph{preserves supremums}, namely, for any $x \in \vV$ and $S \subseteq \vV$, $x \vten \vJoin S = \vJoin \left\{ x \vten a \mid a \in S \right\}$ and $\vJoin S \vten x = \vJoin \left\{ a \vten x \mid a \in S \right\}$.
	\end{itemize}
\end{definition}
We say that $\vV$ is \emph{unital} if it contains an element $\vun \in \vV$ that is neutral for $\vten$ ($x \vten \vun = x = \vun \vten x$ for all $x \in \vV$).
In this work, we will assume that quantales are unital. The main example we consider is the extended nonnegative reals, also called Lawvere's quantale~(\!\!\cite{Lawvere73}).
\begin{example}
	Nonnegative real numbers with $\infty$ form a complete lattice. By defining addition with infinity as $x + \infty = \infty$ for all $x \in [0,\infty]$, we obtain an associative operation that preserves infimums. Thus, $[0,\infty]$ with the opposite order $\geq$ and addition forms a quantale with unit $0$, denoted $\zeroinfQ$.
\end{example}

\begin{definition}
	A \emph{$\vV$-relation} is a set $X$ equipped with a function $d: X\times X \rightarrow \vV$, often denoted as a pair $(X,d)$. A morphism between $\vV$-relations $(X,d_X)$ and $(Y,d_Y)$ is a function $f: X \rightarrow Y$ satisfying, for all $x,x' \in X$, $d_X(x,x') \vle d_Y(f(x),f(x'))$. We denote by $\vRelV$ the category of $\vV$-relations, where identities and compositions are as for functions.
\end{definition}
An isomorphism between $\vV$-relations $(X,d_X)$ and $(Y,d_Y)$ is a function $f: X \rightarrow Y$ that is a bijection and satisfies, for all $x,x' \in X$, $d_X(x,x') = d_Y(f(x),f(x'))$.
\begin{example}
	A $\zeroinfQ$-relation $(X,d)$ is often thought of as a set equipped with a \emph{distance}. For $x,y \in X$, $d(x,y)$ is a number that could represent the physical distance, or time, or amount of work, required to go from $x$ to $y$. A morphism of $\zeroinfQ$-relations is called a \emph{nonexpansive map} because $d_X(x,x') \geq d_Y(f(x),f(x'))$ is interpreted as $f$ mapping elements closer together (not expanding the space). An isomorphism of $\zeroinfQ$-relations is also called an \emph{isometry}.
\end{example}
The work on quantitative algebraic reasoning initiated in~\cite{Mardare2016} relies on representing distances as families of binary predicates. More precisely, given $\varepsilon \in \vV$, the \emph{equality up to~$\varepsilon$} predicate on a $\vV$-relation $(X,d)$, denoted by $=_{\varepsilon}$, is defined as $x =_\varepsilon y \Leftrightarrow d(x,y) \vge \varepsilon$. For $\zeroinfQ$, we read $x=_{\varepsilon} y$ as saying that the distance from $x$ to $y$ is less than $\varepsilon$. These predicates determine $d$, which enables reasoning about $\vV$-relations in a logical system using these predicates, after enforcing monotonicity and continuity conditions.
\begin{proposition}\!\!\!\!\cite[Proposition~2.21]{Sarkis2024}\label{prop:vrelpreds}
	There is a bijective correspondence between $\vV$-relations on a set $X$ and families of relations ${=_{\varepsilon}} \subseteq X\times X$, indexed by $\varepsilon \in \vV$, that are \emph{monotone} (if $\varepsilon \vge \varepsilon'$ and $x =_{\varepsilon} y$ then $x=_{\varepsilon'} y$) and \emph{continuous} (if $x=_{\varepsilon_i} y$ for all $i$ in some index set $I$, then $x=_{\vjoin_i \varepsilon_i} y$).
\end{proposition}
The category $\vRelV$ has categorical products that lift cartesian products of sets. Given two $\vV$-relations $(X,d_X)$ and $(Y,d_Y)$, the $\vV$-relation on $X\times Y$ defined by $d((x,y),(x',y')) = d_X(x,x') \vmeet d_Y(y,y')$ satisfies the universal property of the product. This defines a symmetric monoidal structure on $\vRelV$, which we denote by $\tensorvmeet$. Its unit is the reflexive $\vV$-relation on a singleton, $\mathbf{1}$. 
We may consider categories enriched in $(\vRelV,\tensorvmeet,\mathbf{1})$, but in the following section, we will use another monoidal structure.

Given $(X,d_X)$ and $(Y,d_Y)$, there is a $\vV$-relation on $X \times Y$ defined by $d((x,y),(x',y')) = d_X(x,x') \vten d_Y(y,y')$. When $\vV$ is commutative (i.e.\@ $a \vten b = b \vten a$), this also defines a symmetric monoidal structure on $\vRelV$, which we denote by $\tensorvten$. Its unit is also $\mathbf{1}$. In the sequel, we assume $\vV$ is commutative.

\subsection{Categories Enriched over \texorpdfstring{$\vRelV$}{VRel}}
Enriched $\vRelV$-categories enjoy a straightforward concrete characterisation that will be more convenient to us than the abstract definition of enrichment in e.g.\@ \cite{KellyBook}. Indeed, a $(\vRelV,\tensorvten,\mathbf{1})$-enriched category $\catC$ is simply a category $\catC$ equipped with a $\vV$-relation on each of its homsets where composition is a morphism of $\vV$-relations. Unrolled: for every objects $X,Y \in \catC$, there is a function $d^{\scriptscriptstyle \catC}_{X,Y} : \catC(X,Y) \times \catC(X,Y) \rightarrow \vV$, we often simply write $d^{\scriptscriptstyle \catC}$ or $d$, and for every morphisms $f,f': X \rightarrow Y$ and $g,g': Y \rightarrow Z$ in $\catC$,
\begin{equation}\label{eqn:nexpcomposition}
	d^{\scriptscriptstyle \catC}_{X,Z}(f\scomp g,f'\scomp g') \vge d^{\scriptscriptstyle \catC}_{X,Y}(f,f') \vten d^{\scriptscriptstyle \catC}_{Y,Z}(g,g').
\end{equation}
We can further impose on a $\vRelV$-category that it is monoidally enriched, or that it is an enriched monoidal category, meaning that it has a monoidal structure compatible with the enrichment. Concretely, if $(\catC,\tensor,I)$ is a symmetric monoidal category and $\catC$ is enriched over $\vRelV$, then we say it is \emph{monoidally enriched} if $\tensor$ is a morphism of $\vV$-relations, i.e.\@ for all $f,f':X \rightarrow Y$ and $g,g':X' \rightarrow Y'$,
\begin{equation}\label{eqn:nexptensor}
	d^{\scriptscriptstyle \catC}_{X\tensor X', Y\tensor Y'}(f\tensor g,f'\tensor g') \vge d^{\scriptscriptstyle \catC}_{X,Y}(f,f') \vten d^{\scriptscriptstyle \catC}_{X',Y'}(g,g').
\end{equation}
We will call such a category a \emph{$\vRelV$-SMC}.

We also mention that $\vRelV$-enriched functors are functors between the underlying categories that are locally morphisms of $\vV$-relations, that is, $F: \catC \rightarrow \catD$ is enriched whenever $d^{\scriptscriptstyle\catC}(f,g) \vle d^{\scriptscriptstyle\catD}(F(f),F(g))$ for all morphisms $f,g \in \catC(X,Y)$. Similarly, an enriched isomorphism is an isomorphism between the underlying categories that is locally an isomorphism of $\vV$-relations.

In case $\vV = \zeroinfQ$, we say an enriched functor is \emph{locally nonexpansive} and an enriched isomorphism is \emph{locally an isometry}.

\subsection{Relative Entropy}
Let $\Dset X$ denote the set of \emph{finitely supported distributions} over a set $X$, namely,
\[\Dset X \coloneq \{ \dist : X \rightarrow [0,1] \mid \dist^{-1}(0,1] \text{ is finite and } \textstyle\sum_{x \in X} \dist(x) = 1 \}. \]
We write $\dirac{x} \in \Dset X$ for the \emph{Dirac distribution} at $x \in X$ defined by $\dirac{x}(x') = 0$ if $x \neq x'$ and $\dirac{x}(x) = 1$.

Relative entropy, also called Kullback--Leibler (KL) divergence~(\!\!\cite{Kullback1951}), is a $\zeroinfQ$-relation on $\Dset X$.
\begin{definition}\label{defn:kl}
	\emph{KL divergence} is a map $\Dset X \times \Dset X \rightarrow \zeroinfQ$ defined by $\kl(\dist,\distb) \coloneq \sum_{x \in X} \dist(x) \log\frac{\dist(x)}{\distb(x)}$, with $\log\frac{x}{0} = \infty$ and $0\cdot \infty = 0$ as conventions. In particular, $\kl(\dist,\distb) = \infty$ if and only if there exists some $x \in X$ such that $\dist(x) > 0 = \distb(x)$, and $\kl(\dist,\dist) = 0$.
\end{definition}
\begin{example}
	For any natural number $n \in \N$, let $\fset{n} \coloneq \{1,\dots, n\}$. We can represent a distribution in $\Dset \fset{n}$ as a vector in $[0,1]^n$ whose entries sum up to $1$. Then, the KL divergence between two distributions is
	\[\kl([a_1,\dots,a_n],[b_1,\dots,b_n]) = \textstyle\sum_{i =1}^n a_i\log\frac{a_i}{b_i}.\]
	The following instance, with $n = 2$, will appear several times in the rest of the paper.
	\begin{equation*}
		\kl(\raisebox{0.3ex}{$\distvect{p}{1-p}$},\raisebox{0.3ex}{$\distvect{q}{1-q}$}) = p\log\frac{p}{q} + (1-p)\log\frac{1-p}{1-q}.
	\end{equation*}
	Note that if $q\in \{0,1\}$ and $p \neq q$, then $\kl(\raisebox{0.3ex}{$\distvect{p}{1-p}$},\raisebox{0.3ex}{$\distvect{q}{1-q}$}) = \infty$. If $p = q$, then $\kl(\raisebox{0.3ex}{$\distvect{p}{1-p}$},\raisebox{0.3ex}{\scalebox{0.45}{$\begin{bmatrix} q\\1-q \end{bmatrix}$}}) = 0$.
\end{example}
A renowned property of KL divergence, called the \emph{chain rule}~(\!\!\cite[Theorem~2.15]{Polyanskiy2025}), allows to decompose the computation of the divergence for distributions over products or coproducts of sets. We state two specific instances of the chain rule that will be key in the axiomatisations of \cref{sec:axiombstoch,sec:axiomfstoch}.

\begin{lemma}\label{lem:klprodchainrule}
	Let $Y \neq \emptyset$ be a set, $\B \coloneq \{\b0,\b1\}$, $\dist,\distb \in \Dset(\B \times Y)$, $p \coloneq \sum_{y \in Y} \dist(\b1,y)$, and $q \coloneq \sum_{y \in Y} \distb(\b1,y)$. Writing $\overline{\dist}(\b1,-) \coloneq y \mapsto \sfrac{\dist(\b1,y)}{p}$ and $\overline{\dist}(\b0,-) \coloneq y \mapsto \sfrac{\dist(\b0,y)}{1-p}$ for the conditional distributions on $Y$,\footnote{If $p=0$ (resp.\@ $p=1$) $\overline{\dist}(\b1,-)$ (resp.\@ $\overline{\dist}(\b0,-)$) is assumed to be an arbitrary distribution on $Y$.} and similarly for $\overline{\distb}$, we have
	\begin{equation}\label{eqn:klprodrecurse}
		\kl(\dist,\distb) = \kl(\raisebox{0.3ex}{$\distvect{p}{1-p}$},\raisebox{0.3ex}{$\distvect{q}{1-q}$})+ p\kl(\overline{\dist}(\b1,-),\overline{\distb}(\b1,-)) + (1-p)\kl(\overline{\dist}(\b0,-),\overline{\distb}(\b0,-)).
	\end{equation}
\end{lemma}

\begin{lemma}\label{lem:klsumchainrule}
	Let $X,Y \neq \emptyset$ be sets, $\dist, \distb \in \Dset(X+Y)$, $p \coloneq \sum_{x \in X}\dist(x)$ and $q \coloneq \sum_{x \in X}\distb(x)$. Writing $\dist_X \coloneq x \mapsto \dist(x)/p$ for the conditional distribution on $X$, and similarly for $\dist_Y$, $\distb_X$, and $\distb_Y$, we have
	\begin{equation}\label{eqn:klsumchainrule}
		\kl(\dist,\distb) = \kl(\raisebox{0.3ex}{$\distvect{p}{1-p}$},\raisebox{0.3ex}{$\distvect{q}{1-q}$}) + p\kl(\dist_X,\distb_X) + (1-p)\kl(\dist_Y,\distb_Y).
	\end{equation}
\end{lemma}
We remark that, in both cases, $\kl(\dist,\distb)$ is a monotone function of the divergences between the conditional distributions. More precisely, looking at \cref{eqn:klprodrecurse}, if you bound $\kl(\overline{\dist}(\b{x},-),\overline{\distb}(\b{x},-))$ from above, then you get an upper bound for $\kl(\dist,\distb)$, and similarly for \cref{eqn:klsumchainrule}:
\begin{align}
	\label{eqn:monotoneklprod}
	\kl(\overline{\dist}(\b1,-),\overline{\distb}(\b1,-)) \leq \varepsilon \text{ and } \kl(\overline{\dist}(\b0,-),\overline{\distb}(\b0,-)) \leq \delta & \implies \kl(\dist,\distb) \leq \kl(\raisebox{0.3ex}{$\distvect{p}{1-p}$},\raisebox{0.3ex}{$\distvect{q}{1-q}$}) + p\varepsilon + (1-p)\delta  \\
	\label{eqn:monotoneklsum}
	\kl(\dist_X,\distb_X) \leq \varepsilon \text{ and } \kl(\dist_Y,\distb_Y) \leq \delta                                                                 & \implies \kl(\dist,\distb) \leq \kl(\raisebox{0.3ex}{$\distvect{p}{1-p}$},\raisebox{0.3ex}{$\distvect{q}{1-q}$}) + p\varepsilon + (1-p)\delta.
\end{align}
This is particularly relevant for us because our logical framework does not reason about exact distance. We use the predicates $=_{\varepsilon}$ which are interpreted as upper bounds on the distance by $\varepsilon \in \zeroinfQ$.

\section{\texorpdfstring{$\vRelV$}{VRel}-Enriched Monoidal Theories}\label{sec:enrichmonth}

The framework of quantitative monoidal algebra, as developed in~\cite{Lobbia2025}, only allows for inference rules describing the quantitative behaviour of sequential composition and parallel composition of string diagrams (on top of the usual pseudometric axioms like triangle inequality). In order to axiomatise relative entropy in \cref{sec:axkl,sec:axren} below, we need a more general setting, allowing to model the chain rule as an inference rule of our diagrammatic theory. For this purpose, we now develop an extension of quantitative monoidal algebra, which allows the formation of inference rules not necessarily restricted to sequential and parallel composition of string diagrams. Besides contributing to our later results, this study is of interest for other works on quantitative algebra, as discussed in Section~\ref{sec:conclusions}.

\subsection{Quantitative Implications}
Let $\sig$ be a monoidal signature and $\vync{\sig}$ be its set of terms. We define \emph{quantitative implications}.
\begin{definition}
	A \emph{$\vV$-equation} between terms is a pair of terms $(s,t)$ in $\vync{\sig}$ along with an element $\varepsilon \in \vV$, which we shall write $s=_{\varepsilon} t$. 
	A \emph{$\vV$-implication} consists of a set $\Gamma$ of ($\vV$-)equations called the \emph{premises} along with a ($\vV$-)equation $\eqn$ called the \emph{conclusion}, which we will write as $\eqns \Rightarrow \eqn$.

	Given a $\vRelV$-SMC $\catC$ and a symmetric monoidal functor $\sem{-}: \vync{\sig} \rightarrow \catC$, we define \emph{satisfaction}. An equation $f=g$ is satisfied when $\sem{f} = \sem{g}$. A $\vV$-equation $f=_{\varepsilon} g$ is satisfied when $d^{\scriptscriptstyle \catC}(\sem{f},\sem{g}) \vge \varepsilon$. A $\vV$-implication $\eqns \Rightarrow \eqn$ is satisfied when either $\eqn$ is satisfied or at least one ($\vV$)-equation in $\Gamma$ is not satisfied.
\end{definition}
Oftentimes, we write an expression that is interpreted as a family of $\vV$-implications. For instance, 
we can write the following $\vV$-implications that are always satisfied in $\vRelV$-SMCs by \cref{prop:vrelpreds}.
\begin{align}
	\input{implicationmon.tikz}  &  & \forall \varepsilon,\varepsilon' \in \vV \text{ such that } \varepsilon\vge \varepsilon'\label{eqn:implicationmon} \\
	\input{implicationcont.tikz} &  & \forall \left\{ \varepsilon_i \right\}_{i \in I} \subseteq \vV\label{eqn:implicationcont}
\end{align}
There is an explicit quantification over elements in $\vV$, but there is another implicit quantification of $f$ and $g$:
$f$ and $g$ are \emph{metavariables}, and we interpret these expressions as families of $\vV$-implications where $f$ and $g$ are replaced by every possible terms that typecheck (e.g.~$f$ and $g$ must have the same inputs and outputs to be comparable in \cref{eqn:implicationmon,eqn:implicationcont}).

In the sequel, we will denote metavariables with symbols that are not defined in the scope to implicitly imply that they are universally quantified. We often write $f$, $g$, $h$, etc.\@ for generic terms, and $\varepsilon$, $\delta$, etc.\@ for generic elements of $\vV$.

\begin{example}\label{exmp:validimplications}
	Let us continue the list of $\vV$-implications that are \emph{valid}, i.e.\@ satisfied by any functor to a $\vRelV$-SMC. Recall that \cref{eqn:implicationmon,eqn:implicationcont} are satisfied by virtue of every hom-set of a $\vRelV$-SMC being a $\vV$-relation (along with \cref{prop:vrelpreds}).
	\begin{itemize}[nosep]
		\item The following implication is valid because sequential composition is a $\vV$-relation morphism \cref{eqn:nexpcomposition}.
		      \begin{equation}\label{eqn:implicationnexpseq}
			      \input{implicationnexpseq.tikz}
		      \end{equation}
		\item The following implication is valid because parallel composition is a $\vV$-relation morphism \cref{eqn:nexptensor}.
		      \begin{equation}\label{eqn:implicationnexppar}
			      \input{implicationnexppar.tikz}
		      \end{equation}
		\item The following implications are valid because equality is an equivalence relation congruent with respect to both compositions and to the $\vV$-relations on homsets.
		      \begin{equation}\label{eqn:equality}
			      \begin{array}{c}
				      \emptyset \Rightarrow f=f \qquad f=g \Rightarrow g=f \qquad f=g,g=h \Rightarrow f=h \\ f=f', g=g' \Rightarrow f\scomp g = f'\scomp g' \qquad f=f', g=g' \Rightarrow f\tensor g = f' \tensor g'\\
				      f=f',f=_{\varepsilon} g \Rightarrow  f'=_{\varepsilon} g \qquad  g=g', f=_{\varepsilon} g \Rightarrow f =_{\varepsilon} g'.
			      \end{array}
		      \end{equation}
		\item The equations in \cref{fig:axiomssmc}, seen as implications with no premises, are all valid because $\vRelV$-SMCs satisfy the laws of SMCs.
	\end{itemize}
\end{example}
We often omit writing $\emptyset$ when a $\vV$-implication has no premise.
\begin{remark}
	We consider SMCs whose homsets are mere $\vV$-relations and not metrics, so the list of valid implications does not include $\ \Rightarrow f=_0f$, $f=_0 g \Rightarrow f = g$, nor $f=_{\varepsilon} g, g =_{\delta} h \Rightarrow f =_{\varepsilon \vten \delta} g$. Each or all of these can be asserted in a case by case basis (e.g.~$\ \Rightarrow f=_0f$ is an axiom of the theories defined in \cref{defn:thklprod,defn:thklsum,defn:threprod,defn:thresum}).
\end{remark}
\begin{definition}
	A \emph{$\vV$-theory} $\th$ consists of a monoidal signature $\sig$ and a set of $\vV$-implications over $\sig$ called \emph{axioms}. A \emph{model} of $\th$ is a $\vRelV$-SMC equipped with a symmetric monoidal functor $\vync{\sig} \rightarrow \catC$ that satisfies all the implications in $\th$.
\end{definition}

A theory can be closed under standard implicational reasoning to contain all additional implications that are satisfied by all its models. For example, for any ($\vV$-)equation $\eqn$, all models trivially satisfy $\eqn \Rightarrow \eqn$, so that implication should be in the closure of $\th$. More generally, the closure is a consequence relation.
\begin{definition}
	Let $\vdash$ denote a set of $\vV$-implications and write $\eqns \Rightarrow \eqn \in {\vdash}$ infix as $\eqns \vdash \eqn$. We say that $\vdash$ is a \emph{consequence relation} if the following holds.
	\begin{enumerate}[nosep]
		\item $\eqn \in \eqns$ implies $\eqns \vdash \eqn$.
		\item $\eqns \vdash \eqn$ and $\eqns \subseteq \eqns'$ implies $\eqns' \vdash \eqn$.
		\item If $\eqns \vdash \eqn$ for all $\eqn \in \eqnsb$ and $\eqnsb\cup \eqnsb' \vdash \eqnb$, then $\eqns\cup \eqnsb' \vdash \eqnb$.
	\end{enumerate}
\end{definition}

Moreover, as we said in \cref{exmp:validimplications}, \cref{eqn:implicationmon,eqn:implicationcont,eqn:implicationnexpseq,eqn:implicationnexppar,eqn:equality} and \cref{fig:axiomssmc} are always satisfied, so they must also be put in the closure of $\th$. Perhaps unsurprisingly, the laws of consequence relations, \cref{eqn:implicationmon,eqn:implicationcont,eqn:implicationnexpseq,eqn:implicationnexppar,eqn:equality}, and \cref{fig:axiomssmc} are necessary and sufficient to construct the set of \emph{all} $\vV$-implications satisfied by all models of $\th$.

\begin{definition}
	Given a theory $\th$, let $\vdash^{\th}$ denote the smallest consequence relation containing $\th$ and all the $\vV$-implications in \cref{eqn:implicationmon,eqn:implicationcont,eqn:implicationnexpseq,eqn:implicationnexppar,eqn:equality} and \cref{fig:axiomssmc}. 
	We call it the \emph{closure} of $\th$.
\end{definition}

\begin{remark}
	It can be more convenient to write a $\vV$-implication as an inference rule:
	\[
		\eqn_1, \dots, \eqn_n \Rightarrow \eqn \qquad\rightsquigarrow\qquad \begin{bprooftree}
			\AxiomC{$\eqn_1$}
			\AxiomC{$...$}
			\AxiomC{$\eqn_n$}
			\TrinaryInfC{$\eqn$}
		\end{bprooftree}.
	\]
	Then, one can show that $\eqns \vdash^{\th} \eqn$ if and only if there is a proof tree with conclusion $\eqn$ and inference rules corresponding to implications in $\th$, \cref{eqn:implicationmon,eqn:implicationcont,eqn:implicationnexpseq,eqn:implicationnexppar,eqn:equality}, or \cref{fig:axiomssmc} such that all leaves are in $\eqns$ or empty. The laws of consequence relations are reflected in how proof trees are constructed.
\end{remark}
In order to prove our claim that $\vdash^{\th}$ contains all implications satisfied by all models, we construct the initial model as follows.
\begin{definition}\label{defn:synvcat}
	Given a monoidal signature $\sig$, a $\vV$-theory $\th$, and its closure $\vdash^{\th}$, let the \emph{syntactic $\vRelV$-SMC}, denoted $\sync{\th}$ be defined by
	\begin{itemize}[nosep]
		\item Its objects are lists of sorts in $\sig_0$.
		\item Its morphisms are terms in $\vync{\sig}$ quotiented by the relation $\equiv$ defined by $f \equiv g \Leftrightarrow {\vdash^{\th} f = g}$.
		\item The $\vV$-relation $d^{\th}$ on homsets is defined by $d^{\th}(f,g) = \vJoin\left\{ \varepsilon \mid {\vdash^{\th} f=_{\varepsilon} g} \right\}$.
	\end{itemize}
\end{definition}
We prove that this is a valid definition of a $\vRelV$-SMC and that it is the initial model of $\th$ at once.

\begin{lemma}\label{lem:syntactic}
	The canonical quotient map $\vync{\sig} \rightarrow \sync{\th}$ is a model of $\th$, and any model $\vync{\sig} \rightarrow \catC$ of $\th$  uniquely factorises through an enriched symmetric monoidal functor $\sync{\th} \rightarrow \catC$.
\end{lemma}
\begin{proof}
	The terms $f$ and $g$ are equated in $\sync{\th}$ whenever there is a proof of $f=g$ using the inference rules in $\th$ and \cref{eqn:implicationmon,eqn:implicationcont,eqn:implicationnexpseq,eqn:implicationnexppar,eqn:equality} whose leaves are empty. It follows from \cref{eqn:equality} and the axioms of consequence relations that (i) $\equiv$ is an equivalence relation congruent with respect to the structure of an SMC. Moreover, one can show by structural induction on proofs of $\vdash^{\th} f = g$ that (ii) for any model $\sem{-}:\vync{\sig} \rightarrow \catC$ of $\th$, and any terms $f,g: u \rightarrow v$, $f\equiv g$ implies that $\sem{f} = \sem{g}$.

	Next, $d^{\th}(f,g)$ is the largest $\varepsilon$ for which there is a proof of $f=_{\varepsilon} g$ using the inference rules in $\th$ and \cref{eqn:implicationmon,eqn:implicationcont,eqn:implicationnexpseq,eqn:implicationnexppar,eqn:equality} whose leaves are empty. It follows from \cref{eqn:equality} that (a) $d^{\th}$ is congruent with $\equiv$ and from \cref{eqn:implicationnexpseq,eqn:implicationnexppar} that (b) both compositions are nonexpansive with respect to $d^{\th}$.  Moreover, one can show by structural induction on proofs of $\vdash^{\th} f=_{\varepsilon} g$ that (c) for any model $\sem{-}:\vync{\sig} \rightarrow \catC$ of $\th$, and any terms $f,g: u \rightarrow v$, $d^{\th}(f,g) \vle d^{\scriptscriptstyle \catC}(\sem{f},\sem{g})$.

	By items (i), (a), and (b), and the axioms in \cref{fig:axiomssmc}, the quotient of $\vync{\sig}$ by $\equiv$ equipped with the $\vV$-relations $d^{\th}$ is a well-defined $\vRelV$-SMC, and that the quotient map $\vync{\sig} \rightarrow \sync{\th}$ is a model of $\th$. By items (ii) and (c), any model $\sem{-}: \vync{\sig} \rightarrow \catC$ can be uniquely factored through a $\vRelV$-enriched symmetric monoidal functor $\sync{\th} \rightarrow \catC$.
\end{proof}
We can now show $\vdash^{\th}$ satisfies the desired property of a closure.
\begin{theorem}\label{thm:completeness}
	For any $\vV$-implication $\eqns \Rightarrow \eqn$, $\eqns \vdash^{\th} \eqn$ if and only if all models of $\th$ satisfy $\eqns \Rightarrow \eqn$.
\end{theorem}
\begin{proof}
	The forward implication is shown by structural induction on the proof of $\Gamma \vdash^{\th} \eqn$. It resembles the steps in the proof of \cref{lem:syntactic} showing that models factorise through $\sync{\th}$, but with a non-empty context.

	For the converse direction, assume the $\vV$-implication $\eqns \Rightarrow \eqn$ is satisfied by all models of $\th$. We wish to show $\eqns \vdash^{\th} \eqn$. Let $\th'$ be the theory consisting of the signature $\sig$ and the axioms of $\th$ along with the set of implications $\{\emptyset \Rightarrow \eqnb \mid \eqnb \in \eqns\}$. By definition, $\th'$ is a theory whose models are exactly the models of $\th$ that also satisfy every ($\vV$)-equation in $\eqns$.

	By \cref{lem:syntactic}, the syntactic category $\sync{\th'}$ is a model of $\th'$, and it is in particular a model of $\th$. As $\eqns \Rightarrow \eqn$ is satisfied by all models of $\th$ by hypothesis, it is satisfied by $\sync{\th'}$. Furthermore, by the construction of $\th'$, every ($\vV$)-equation in $\eqns$ is satisfied in $\sync{\th'}$. Hence, the conclusion $\eqn$ must also be satisfied in $\sync{\th'}$, which means there is a derivation of $\vdash^{\th'} \eqn$. This is a proof tree using the axioms of $\th$ or $\eqns$, which we can see as a derivation (in the original theory) of $\eqns \vdash^{\th} \eqn$ as desired.
\end{proof}

\section{Axiomatisation of KL Divergence for Stochastic Matrices}\label{sec:axkl}
The two monoidal theories we recalled in \cref{fig:causcircaxiom,fig:fritzaxiom} allow us to reason about stochastic matrices using string diagrams. The goal of this section is to extend this reasoning with quantitative implications in order to reason diagrammatically about the KL divergence between stochastic matrices.

In \cref{sec:enrichfstoch}, we show that KL divergence is compatible with the categorical structures of stochastic matrices, yielding two $\zeroinfQ$-SMCs: $\BSklprod$ and $\FSklsum$. In \cref{sec:axiombstoch}, we extend the theory $\th_{\!\Circ}$ with $\zeroinfQ$-implications to obtain an axiomatisation of $\BSklprod$. In \cref{sec:axiomfstoch}, we similarly extend $\th_{\!\ConvAlg}$ to obtain an axiomatisation of $\FSklsum$.

\subsection{\texorpdfstring{$\zeroinfQ$}{[0,infty]\_+}-SMCs of Stochastic Matrices}\label{sec:enrichfstoch}
KL divergence~(\cref{defn:kl}) was defined on probability distributions, but we can extend it to stochastic matrices via a column-wise maximum. Indeed, seeing each column of an $m\times n$ stochastic matrix as a distribution in $\Dset \fset{m}$ sending $j$ to the value of the $j$th entry of that column, we define $\klmax$ on $\FStoch(n,m)$ by
\begin{equation}\label{eqn:klmax}
	\klmax(A,B) \coloneq \max_{i \in \fset{n}}\kl(A_i,B_i),
\end{equation}
where $A_i$ (resp.\@ $B_i$) is the $i$th column of $A$ (resp.\@ $B$), and $\kl(A_i,B_i)$ is the KL divergence between the corresponding distributions in $\Dset \fset{m}$. When $m$ or $n$ are zero, we set $\klmax([],[]) \coloneq 0$. Perrone showed in~\cite[Proposition~2.10]{Perrone2024} that this definition satisfies both requirements~(\cref{eqn:nexpcomposition,eqn:nexptensor}) to obtain a $\zeroinfQ$-SMC $\FSklprod$, which is $\FSprod$ equipped with $\klmax$ on each of its homsets.

It readily follows that equipping $\BStoch$ with $\klmax$ on its homsets yields another $\zeroinfQ$-SMC that we write $\BSklprod$. By viewing $2^m\times 2^n$ stochastic matrices as functions taking bit strings in $\B^n$ and returning distributions in $\Dset(\B^m)$, we can rewrite \cref{eqn:klmax} as
\begin{equation}\label{eqn:klmaxalt}
	\klmax(A,B) = \max_{u \in \B^n}\kl(A(u),B(u)).
\end{equation}
Furthermore, we can state the chain rule in \cref{lem:klprodchainrule} when $Y = \B^n$ with string diagrams since distributions in $\Dset(\B^n)$ correspond, via the isomorphism $\sem{-}$ defined in \cref{eqn:sembsprod}, to morphisms in $\sync{\Circ}(0,n)$. For any morphisms $f_{\b1},f_{\b0},g_{\b1},g_{\b0} \in \sync{\Circ}(0,n)$, and any $p,q \in [0,1]$, we have (omitting applications of $\sem{-}$)
\begin{equation}\label{eqn:chainproddiags}
	\klmax(\scalebox{0.7}{\input{distdecompose.tikz}},\scalebox{0.7}{\input{distsbdecompose.tikz}}) = \kl(\raisebox{0.3ex}{$\distvect{p}{1-p}$},\raisebox{0.3ex}{$\distvect{q}{1-q}$})+p\klmax(\scalebox{0.7}{$\statediag{f_{\b1}}{}$},\scalebox{0.7}{$\statediag{g_{\b1}}{}$}) + (1-p)\klmax(\scalebox{0.7}{$\statediag{f_{\b0}}{}$},\scalebox{0.7}{$\statediag{g_{\b0}}{}$})
\end{equation}
This is exactly \cref{eqn:klprodrecurse} after noting that when $f_{\b{x}}$ corresponds to $\overline{\dist}(\b{x},-)$, $\scalebox{0.7}{\input{distdecompose.tikz}}$ corresponds to $\dist$. Indeed, with the first step being an instance of \cref{eqn:semconvex} (with only one column), we have
\[\sem{\scalebox{0.7}{\input{distdecompose.tikz}}} = p\sem{\flip{\b1} \otimes f_{\b1}} + (1-p)\sem{\flip{\b0}\otimes f_{\b0}} = p(\ket{\b1}\otimes \overline{\dist}(\b1,-)) + (1-p)(\ket{\b0}\otimes \overline{\dist}(\b0,-)) = \dist\]
We have a similar decomposition for $g$ and $\distb$.

We record for later another property of $\BSklprod$: for any morphisms $f_{\b1},g_{\b1} : n \rightarrow m$ and $f_{\b0},g_{\b0}: n' \rightarrow m$ in $\BSklprod$, we have
\begin{equation}\label{eqn:ifmaxdiags}
	\klmax(\scalebox{0.7}{$\input{iff.tikz}$},\scalebox{0.7}{$\input{ifg.tikz}$}) = \max \left\{ \klmax(\scalebox{0.7}{$\begin{tikzpicture}
	\begin{pgfonlayer}{nodelayer}
		\node [style=basic box] (24) at (0, 0) {$f_{\b1}$};
		\node [style=none] (25) at (1, 0) {};
		\node [style=none] (27) at (-1, 0) {};
	\end{pgfonlayer}
	\begin{pgfonlayer}{edgelayer}
		\draw (24) to (25.center);
		\draw (27.center) to (24);
	\end{pgfonlayer}
\end{tikzpicture}
$},\scalebox{0.7}{$\begin{tikzpicture}
	\begin{pgfonlayer}{nodelayer}
		\node [style=basic box] (24) at (0, 0) {$g_{\b1}$};
		\node [style=none] (25) at (1, 0) {};
		\node [style=none] (27) at (-1, 0) {};
	\end{pgfonlayer}
	\begin{pgfonlayer}{edgelayer}
		\draw (24) to (25.center);
		\draw (27.center) to (24);
	\end{pgfonlayer}
\end{tikzpicture}
$}), \klmax(\scalebox{0.7}{$\begin{tikzpicture}
	\begin{pgfonlayer}{nodelayer}
		\node [style=basic box] (24) at (0, 0) {$f_{\b0}$};
		\node [style=none] (25) at (1, 0) {};
		\node [style=none] (27) at (-1, 0) {};
	\end{pgfonlayer}
	\begin{pgfonlayer}{edgelayer}
		\draw (24) to (25.center);
		\draw (27.center) to (24);
	\end{pgfonlayer}
\end{tikzpicture}
$},\scalebox{0.7}{$\begin{tikzpicture}
	\begin{pgfonlayer}{nodelayer}
		\node [style=basic box] (24) at (0, 0) {$g_{\b0}$};
		\node [style=none] (25) at (1, 0) {};
		\node [style=none] (27) at (-1, 0) {};
	\end{pgfonlayer}
	\begin{pgfonlayer}{edgelayer}
		\draw (24) to (25.center);
		\draw (27.center) to (24);
	\end{pgfonlayer}
\end{tikzpicture}
$}) \right\}.
\end{equation}
\begin{proof*}{Proof of \cref{eqn:ifmaxdiags}.}
	Let $A$ be the matrix corresponding to \raisebox{0.5ex}{\scalebox{0.7}{$\input{iff.tikz}$}} and $B$ correspond to \raisebox{0.5ex}{\scalebox{0.7}{$\input{ifg.tikz}$}}. Any bit string in $\B^{1+n+n'}$ can be written as $\b{x}u_{\b1}u_{\b0}$ with $\b{x}\in \B$, $u_{\b1} \in \B^n$, and $u_{\b0} \in \B^{n'}$ so that $A(\b{x}u_{\b1}u_{\b0}) =\sem{f_\b{x}}(u_{\b{x}})$ and $B(\b{x}u_{\b1}u_{\b0}) =\sem{g_\b{x}}(u_{\b{x}})$.
	Thus,
	\begin{align*}
		\klmax(A,B) & = \max_{\b{x}, u_{\b1}, u_{\b0}} \kl(A(\b{x}u_{\b1}u_{\b0}),B(\b{x}u_{\b1}u_{\b0}))                                                                                                                           &  & \text{by \cref{eqn:klmaxalt}} \\
		            & = \max \{ \max_{u_{\b1}} \kl(\sem{f_{\b1}}(u_{\b1}),\sem{g_{\b1}}(u_{\b1})), \max_{u_{\b0}} \kl(\sem{f_{\b0}}(u_{\b0}),\sem{g_{\b0}}(u_{\b0})) \}                                                                                                \\
		            & =\max \{ \klmax(\scalebox{0.7}{$\begin{tikzpicture}
	\begin{pgfonlayer}{nodelayer}
		\node [style=basic box] (24) at (0, 0) {$f_{\b1}$};
		\node [style=none] (25) at (1, 0) {};
		\node [style=none] (27) at (-1, 0) {};
	\end{pgfonlayer}
	\begin{pgfonlayer}{edgelayer}
		\draw (24) to (25.center);
		\draw (27.center) to (24);
	\end{pgfonlayer}
\end{tikzpicture}
$},\scalebox{0.7}{$\begin{tikzpicture}
	\begin{pgfonlayer}{nodelayer}
		\node [style=basic box] (24) at (0, 0) {$g_{\b1}$};
		\node [style=none] (25) at (1, 0) {};
		\node [style=none] (27) at (-1, 0) {};
	\end{pgfonlayer}
	\begin{pgfonlayer}{edgelayer}
		\draw (24) to (25.center);
		\draw (27.center) to (24);
	\end{pgfonlayer}
\end{tikzpicture}
$}), \klmax(\scalebox{0.7}{$\begin{tikzpicture}
	\begin{pgfonlayer}{nodelayer}
		\node [style=basic box] (24) at (0, 0) {$f_{\b0}$};
		\node [style=none] (25) at (1, 0) {};
		\node [style=none] (27) at (-1, 0) {};
	\end{pgfonlayer}
	\begin{pgfonlayer}{edgelayer}
		\draw (24) to (25.center);
		\draw (27.center) to (24);
	\end{pgfonlayer}
\end{tikzpicture}
$},\scalebox{0.7}{$\begin{tikzpicture}
	\begin{pgfonlayer}{nodelayer}
		\node [style=basic box] (24) at (0, 0) {$g_{\b0}$};
		\node [style=none] (25) at (1, 0) {};
		\node [style=none] (27) at (-1, 0) {};
	\end{pgfonlayer}
	\begin{pgfonlayer}{edgelayer}
		\draw (24) to (25.center);
		\draw (27.center) to (24);
	\end{pgfonlayer}
\end{tikzpicture}
$}) \}. &  & \text{by \cref{eqn:klmaxalt}}
	\end{align*}

	\vspace{-1\baselineskip}
	\vspace{-1\baselineskip}
\end{proof*}

We can also consider the direct sum instead of the Kronecker product, and $\klmax$ still satisfies \cref{eqn:nexptensor}:
\begin{align*}
	\klmax(A \oplus A', B\oplus B') & = \max_{i} \kl((A\oplus A')_i,(B\oplus B')_i)                \\
	                                & = \max \{ \max_{i} \kl(A_i,B_i) , \max_{j} \kl(A'_j,B'_j) \} \\
	                                & = \max \{ \klmax(A,B), \klmax(A',B') \}                      \\
	                                & \leq \klmax(A,B) + \klmax(A',B'),
\end{align*}
where the second step holds because KL divergence is not affected by elements that are not in the support of the distributions involved. We obtain a $\zeroinfQ$-SMC $\FSklsum$, which is $\FSsum$ equipped with $\klmax$ on its homsets.
We record for later the stronger intermediate step in the derivation above: for any morphisms $f,g: n \rightarrow m$ and $f',g': n' \rightarrow m'$ in $\FSklsum$, we have
\begin{equation}\label{eqn:klnexppar}
	\klmax(\scalebox{0.7}{$\input{parletters.tikz}$},\scalebox{0.7}{$\input{parletterstwo.tikz}$}) \leq \max \left\{ \klmax(\scalebox{0.7}{$\begin{tikzpicture}
	\begin{pgfonlayer}{nodelayer}
		\node [style=basic box] (24) at (0, 0) {$f$};
		\node [style=none] (25) at (1, 0) {};
		\node [style=none] (27) at (-1, 0) {};
	\end{pgfonlayer}
	\begin{pgfonlayer}{edgelayer}
		\draw (24) to (25.center);
		\draw (27.center) to (24);
	\end{pgfonlayer}
\end{tikzpicture}
$},\scalebox{0.7}{$\begin{tikzpicture}
	\begin{pgfonlayer}{nodelayer}
		\node [style=basic box] (24) at (0, 0) {$g$};
		\node [style=none] (25) at (1, 0) {};
		\node [style=none] (27) at (-1, 0) {};
	\end{pgfonlayer}
	\begin{pgfonlayer}{edgelayer}
		\draw (24) to (25.center);
		\draw (27.center) to (24);
	\end{pgfonlayer}
\end{tikzpicture}
$}), \klmax(\scalebox{0.7}{$\begin{tikzpicture}
	\begin{pgfonlayer}{nodelayer}
		\node [style=basic box] (24) at (0, 0) {$f'$};
		\node [style=none] (25) at (1, 0) {};
		\node [style=none] (27) at (-1, 0) {};
	\end{pgfonlayer}
	\begin{pgfonlayer}{edgelayer}
		\draw (24) to (25.center);
		\draw (27.center) to (24);
	\end{pgfonlayer}
\end{tikzpicture}
$},\scalebox{0.7}{$\begin{tikzpicture}
	\begin{pgfonlayer}{nodelayer}
		\node [style=basic box] (24) at (0, 0) {$g'$};
		\node [style=none] (25) at (1, 0) {};
		\node [style=none] (27) at (-1, 0) {};
	\end{pgfonlayer}
	\begin{pgfonlayer}{edgelayer}
		\draw (24) to (25.center);
		\draw (27.center) to (24);
	\end{pgfonlayer}
\end{tikzpicture}
$}) \right\}.
\end{equation}

Analogously to \cref{eqn:chainproddiags}, we can state the chain rule of \cref{lem:klsumchainrule} with string diagrams, except now distributions in $\Dset \fset{m}$ correspond to terms in $\FSsum(1,m)$. For any morphisms $f, g \in \FStoch(1,m)$, $f',g' \in \FStoch(1,m')$, and any $p,q \in [0,1]$, we have
\begin{equation}\label{eqn:chainsumdiags}
	\klmax(\scalebox{0.7}{$\input{decomposedistm.tikz}$},\scalebox{0.7}{$\input{decomposedistbm.tikz}$}) = \kl(\raisebox{0.3ex}{$\distvect{p}{1-p}$},\raisebox{0.3ex}{$\distvect{q}{1-q}$})+p\klmax(\scalebox{0.7}{$\input{distf.tikz}$},\scalebox{0.7}{$\input{distg.tikz}$}) + (1-p)\klmax(\scalebox{0.7}{$\input{distfp.tikz}$},\scalebox{0.7}{$\input{distgp.tikz}$}).
\end{equation}

\subsection{Axiomatisation of \texorpdfstring{$\BSklprod$}{BStoch\^otimes\_kl}}\label{sec:axiombstoch}
In this section, we provide a quantitative monoidal axiomatisation of $\BSklprod$. We rely on the axiomatisation of $\BStoch$ in~\cite{Piedeleu2025b} with the monoidal theory $\th_{\!\Circ}$ recalled in \cref{fig:causcircaxiom}.

\begin{definition}\label{defn:thklprod}
	Let $\sig$ be the signature of $\th_{\!\Circ}$. The theory $\th_{\!\KLprod}$ is defined with the same signature $\sig$ and the following $\zeroinfQ$-implications as axioms.
	\begin{itemize}[nosep]
		\item For any axiom $f=g$ in $\th_{\!\Circ}$, $\th_{\!\KLprod}$ contains $\ \Rightarrow f=g$.
		\item For any $f: n \rightarrow m \in \vync{\sig}$, $\th_{\!\KLprod}$ contains $\ \Rightarrow  f=_0 f$.
		\item For any $f_{\b1},g_{\b1}: 0 \rightarrow m$, $f_{\b0},g_{\b0}: 0 \rightarrow m$, and $p,q \in [0,1]$, $\th_{\!\KLprod}$ contains $\Chainprod$.
		\item For any $f_{\b1},g_{\b1}: n \rightarrow m$ and $f_{\b0},g_{\b0}: n' \rightarrow m$, $\th_{\!\KLprod}$ contains $\Ifmax$.
		      \[\begin{bprooftree}
				      \AxiomC{$\raisebox{0.3ex}{\scalebox{0.7}{$\begin{tikzpicture}
	\begin{pgfonlayer}{nodelayer}
		\node [style=basic box] (24) at (0, 0) {$f_{\b1}$};
		\node [style=none] (25) at (1, 0) {};
	\end{pgfonlayer}
	\begin{pgfonlayer}{edgelayer}
		\draw [style=double] (24) to (25.center);
	\end{pgfonlayer}
\end{tikzpicture}
$}} =_{\varepsilon} \raisebox{0.3ex}{\scalebox{0.7}{$\begin{tikzpicture}
	\begin{pgfonlayer}{nodelayer}
		\node [style=basic box] (24) at (0, 0) {$g_{\b1}$};
		\node [style=none] (25) at (1, 0) {};
	\end{pgfonlayer}
	\begin{pgfonlayer}{edgelayer}
		\draw [style=double] (24) to (25.center);
	\end{pgfonlayer}
\end{tikzpicture}
$}}$}
				      \AxiomC{$\raisebox{0.3ex}{\scalebox{0.7}{$\begin{tikzpicture}
	\begin{pgfonlayer}{nodelayer}
		\node [style=basic box] (24) at (0, 0) {$f_{\b0}$};
		\node [style=none] (25) at (1, 0) {};
	\end{pgfonlayer}
	\begin{pgfonlayer}{edgelayer}
		\draw [style=double] (24) to (25.center);
	\end{pgfonlayer}
\end{tikzpicture}
$}} =_{\delta} \raisebox{0.3ex}{\scalebox{0.7}{$\begin{tikzpicture}
	\begin{pgfonlayer}{nodelayer}
		\node [style=basic box] (24) at (0, 0) {$g_{\b0}$};
		\node [style=none] (25) at (1, 0) {};
	\end{pgfonlayer}
	\begin{pgfonlayer}{edgelayer}
		\draw [style=double] (24) to (25.center);
	\end{pgfonlayer}
\end{tikzpicture}
$}}$}
				      \RightLabel{\textnormal{$\Chainprod$}}
				      \BinaryInfC{$\scalebox{0.7}{$\input{distdecompose.tikz}$} =_{\kl(\raisebox{0.3ex}{$\distvect{p}{1-p}$},\raisebox{0.3ex}{$\distvect{q}{1-q}$})+p\varepsilon + (1-p)\delta}\scalebox{0.7}{$\input{distsbdecompose.tikz}$}$}
			      \end{bprooftree} \quad \begin{bprooftree}
				      \AxiomC{$\raisebox{0.3ex}{\scalebox{0.7}{$\begin{tikzpicture}
	\begin{pgfonlayer}{nodelayer}
		\node [style=basic box] (24) at (0, 0) {$f_{\b1}$};
		\node [style=none] (25) at (1, 0) {};
		\node [style=none] (27) at (-1, 0) {};
	\end{pgfonlayer}
	\begin{pgfonlayer}{edgelayer}
		\draw[style=double] (24) to (25.center);
		\draw[style=double] (27.center) to (24);
	\end{pgfonlayer}
\end{tikzpicture}
$}} =_{\varepsilon} \raisebox{0.3ex}{\scalebox{0.7}{$\begin{tikzpicture}
	\begin{pgfonlayer}{nodelayer}
		\node [style=basic box] (24) at (0, 0) {$g_{\b1}$};
		\node [style=none] (25) at (1, 0) {};
		\node [style=none] (27) at (-1, 0) {};
	\end{pgfonlayer}
	\begin{pgfonlayer}{edgelayer}
		\draw [style=double] (24) to (25.center);
		\draw [style=double] (27.center) to (24);
	\end{pgfonlayer}
\end{tikzpicture}
$}}$}
				      \AxiomC{$\raisebox{0.3ex}{\scalebox{0.7}{$\begin{tikzpicture}
	\begin{pgfonlayer}{nodelayer}
		\node [style=basic box] (24) at (0, 0) {$f_{\b0}$};
		\node [style=none] (25) at (1, 0) {};
		\node [style=none] (27) at (-1, 0) {};
	\end{pgfonlayer}
	\begin{pgfonlayer}{edgelayer}
		\draw [style=double] (24) to (25.center);
		\draw [style=double] (27.center) to (24);
	\end{pgfonlayer}
\end{tikzpicture}
$}} =_{\delta} \raisebox{0.3ex}{\scalebox{0.7}{$\begin{tikzpicture}
	\begin{pgfonlayer}{nodelayer}
		\node [style=basic box] (24) at (0, 0) {$g_{\b0}$};
		\node [style=none] (25) at (1, 0) {};
		\node [style=none] (27) at (-1, 0) {};
	\end{pgfonlayer}
	\begin{pgfonlayer}{edgelayer}
		\draw [style=double] (24) to (25.center);
		\draw [style=double] (27.center) to (24);
	\end{pgfonlayer}
\end{tikzpicture}
$}}$}
				      \RightLabel{{\textnormal{$\Ifmax$}}}
				      \BinaryInfC{$\scalebox{0.7}{$\input{iff.tikz}$} =_{\max\left\{ \varepsilon,\delta \right\}} \scalebox{0.7}{$\input{ifg.tikz}$}$}
			      \end{bprooftree}\]
	\end{itemize}
	We write $\sync{\KLprod}$ for the syntactic $\zeroinfQ$-SMC generated by $\th_{\!\KLprod}$ (see \cref{defn:synvcat}).
\end{definition}
Our axiomatisation result amounts to showing that the functor $\sem{-}$ defined in \cref{eqn:sembsprod} yields an isomorphism of enriched categories between $\sync{\KLprod}$ and $\BSklprod$.  First, we show $\sem{-}$ satisfies the axioms in $\th_{\!\KLprod}$. 

\begin{lemma}\label{lem:semprodsat}
	The $\zeroinfQ$-implications in $\th_{\!\KLprod}$ are satisfied by $\sem{-}: \vync{\sig} \rightarrow \BSklprod$.
\end{lemma}
\begin{proof}
	As the underlying SMC of $\BSklprod$ is $\BSprod$, the functor $\sem{-}$ still factors through an isomorphism $\sync{\Circ} \cong \BSprod$ by \cite[Theorem~3.13]{Piedeleu2025b}. It follows that $\sem{f} = \sem{g}$ whenever $f=g$ is in $\th_{\!\Circ}$, hence the implications in the first item of \cref{defn:thklprod} are satisfied. The KL divergence satisfies $\kl(\dist,\dist) = 0$ for any distribution $\dist$, hence $\ \Rightarrow f=_0 f$ is satisfied for any term $f$. The implications represented by the inference rules $\Chainprod$ and $\Ifmax$ are satisfied by \cref{eqn:chainproddiags} and \cref{eqn:ifmaxdiags} respectively. Indeed, while \cref{eqn:chainproddiags} and \cref{eqn:ifmaxdiags} are strict equalities, the inference rules with equalities up to $\varepsilon$ are still valid because the distances in the conclusions are monotone functions of the distances in the premises (recall \cref{eqn:monotoneklprod}).
\end{proof}
By \cref{lem:syntactic}, we obtain an enriched symmetric monoidal functor $\sem{-}: \sync{\KLprod} \rightarrow \BSklprod$. Note that this is simply the isomorphism in the axiomatisation result of~\cite{Piedeleu2025b}, but we confirmed it is enriched. In particular, the underlying SMCs of $\sync{\KLprod}$ and $\BSklprod$ are $\sync{\Circ}$ and $\BSprod$ respectively. For the latter, it is by definition, and for the former, it suffices to note that the only axioms that derive equations between terms in $\th_{\!\KLprod}$ come from $\th_{\!\Circ}$.

It remains to show that $\sem{-}$ is locally an isometry. We first tackle the case of terms $0 \rightarrow m$.
\begin{lemma}\label{lem:prod:isometryondist}
	For any $f,g: 0 \rightarrow m \in \sync{\KLprod}$, if $\varepsilon = \klmax(\sem{f},\sem{g})$, then $\ \Rightarrow  f=_{\varepsilon} g$ is in the closure of $\th_{\!\KLprod}$.
\end{lemma}
\begin{proof}
	We proceed by induction on $m$. For $m = 0$, the result is handled by the axiom $f=_0 f$ because there is a single morphism in $\sync{\KLprod}(0,0)$ whose interpretation is the empty matrix $[]$.

	Let $f, g: 0 \rightarrow m+1$ be two morphisms in $\sync{\Circ}$. By reasoning solely within $\th_{\!\Circ}$, we can decompose $f$ and $g$ as follows
	\[\scalebox{0.7}{$\statediag{f}{}$} = \scalebox{0.7}{$\input{distdecompose.tikz}$} \qquad \qquad \scalebox{0.7}{$\statediag{g}{}$} = \scalebox{0.7}{$\input{distsbdecompose.tikz}$},\]
	where $p$ is the total weight that $\sem{f}$ assigns to all bit strings starting with $\b1$, $q$ is that for $\sem{g}$, and $f_{\b{x}}$ and $g_{\b{x}}$ correspond to the distributions conditioned on the first bit being $\b{x}$ (\textit{c.f.}\@ $\overline{\dist}$ and $\overline{\distb}$ in \cref{lem:klprodchainrule}).
	We rely here on the completeness of $\th_{\Circ}$ with respect to the semantics in $\BSprod$~\cite[Theorem~3.13]{Piedeleu2025b}. Both equations are true in the semantics, so there must be a syntactic proof using the equations in $\th{\Circ}$.

	Our induction hypothesis says that there is a derivation, for each $\b{x} \in \B$, of $f_{\b{x}} =_{\varepsilon_{\b{x}}} g_{\b{x}}$, where $\varepsilon_{\b{x}} = \klmax(\sem{f_{\b{x}}},\sem{g_{\b{x}}})$. By applying the $\Chainprod$ rule, we derive $f=_{\varepsilon} g$ with $\varepsilon = \kl(\raisebox{0.3ex}{$\distvect{p}{1-p}$},\raisebox{0.3ex}{$\distvect{q}{1-q}$})+ p\varepsilon_{\b1} + (1-p)\varepsilon_{\b0}$. By \cref{eqn:klprodrecurse}, we find $\varepsilon = \kl(\dist,\distb)$ where $\dist = p\cdot \dirac{\b1} \otimes \sem{f_{\b1}} + (1-p)\cdot \dirac{\b0}\otimes \sem{f_{\b0}} = \sem{f}$, and similarly $\distb = \sem{g}$. Thus, we have derived $f=_{\klmax(\sem{f},\sem{g})} g$ as desired.
\end{proof}
\begin{remark}
	The proof above implicitly relies on a normal form for diagrams $0 \rightarrow m$. If we apply the same decomposition above to \raisebox{0.3ex}{\scalebox{0.7}{$\statediag{f_{\b1}}{}$}} and \raisebox{0.3ex}{\scalebox{0.7}{$\statediag{f_{\b0}}{}$}}, and continue iterating, we end up writing $f$ as a balanced binary tree of convex combinations (recall \cref{eqn:semconvex}) of the generators $\raisebox{0.3ex}{\scalebox{0.7}{\flip{p}}}$. This tree always has the same branching structure for a fixed $m$, only the coefficients change for different diagrams. For example, diagrams $0 \rightarrow 2$ rewrite to the normal form on the left and diagrams $0 \rightarrow 3$ rewrite to the normal form on the right (all the $p$s vary in $[0,1]$): 
	\[\scalebox{0.9}{\input{nfontwo.tikz}} \qquad \scalebox{0.9}{\input{nfonthree.tikz}}\]
	This is different from the normal form for diagrams $0 \rightarrow m$ in \cite[Proof of Theorem 23]{diGiorgio2025} as their sorted weighted trees are not balanced. We note that our normal form is not unique because some convex combinations might have a coefficient $0$ (resp.\@ $1$), making the diagram on their top (resp.\@ bottom) branch irrelevant, i.e.\@ changing the coefficients on that branch does not affect the semantics of the whole diagram.
\end{remark}
The rest of the proof follows the axiomatisation of total variation distance in~\cite{diGiorgio2025}.
\begin{theorem}\label{thm:axbsprodkl}
	The function $f \mapsto \sem{f}$ is a bijective isometry $\sync{\KLprod}(n,m) \rightarrow \BSklprod(n,m)$.
\end{theorem}
\begin{proof}
	We already know it is a nonexpansive function because $\sem{-}$ is an enriched functor. It is a bijection because $\sem{-}$ is an isomorphism of the underlying SMCs, $\sync{\Circ}$ and $\BSprod$. Thus, we only need to show that $\klmax(\sem{f},\sem{g}) \geq d^{\KLprod}(f,g)$ for any two morphisms $f,g: n \rightarrow m$.

	We proceed by induction on $n$. The base case $n=0$ is handled by \cref{lem:prod:isometryondist}.
	Given, $f,g: n+1 \rightarrow m$, we have the following equations by reasoning in $\th_{\!\Circ}$ (see e.g.~\cite[Equation~17]{Bonchi2026}). 
	\[\scalebox{0.7}{$\input{fnp1.tikz}$} = \scalebox{0.7}{$\input{fshandecompose.tikz}$} \qquad \scalebox{0.7}{$\input{gnp1.tikz}$} = \scalebox{0.7}{$\input{gshandecompose.tikz}$}\]
	Let $f_{\b{x}} = \raisebox{0.3ex}{\scalebox{0.7}{$\input{fbx.tikz}$}}$ and $g_{\b{x}} = \raisebox{0.3ex}{\scalebox{0.7}{$\input{gbx.tikz}$}}$ for $\b{x} \in \B$. Our induction hypothesis says that $f_{\b{x}} =_{\klmax(\sem{f_{\b{x}}},\sem{g_{\b{x}}})} g_{\b{x}}$ is derivable for each $\b{x}$, then the $\Ifmax$ rule derives $f =_{\varepsilon} g$, where $\varepsilon = \max\left\{ \klmax(\sem{f_{\b{1}}},\sem{g_{\b{1}}}),\klmax(\sem{f_{\b{0}}},\sem{g_{\b{0}}}) \right\}$, and finally,
	\begin{align*}
		\klmax(\sem{f},\sem{g}) \stackrel{\cref{eqn:klmaxalt}}{=} \max_{u \in \B^{n+1}} \kl(\sem{f}(u),\sem{g}(u))
		= \max_{\b{x} \in \B} \max_{v \in \B^n} \kl(\sem{f}(\b{x}v),\sem{g}(\b{x}v))
		= \max_{\b{x} \in \B} \klmax(\sem{f_{\b{x}}},\sem{g_{\b{x}}})
		= \varepsilon.
	\end{align*}
	We conclude that $\klmax(\sem{f},\sem{g}) \geq d^{\KLprod}(f,g)$, hence that $\sem{-}$ is locally an isometry on all of $\sync{\KLprod}$.
\end{proof}

\subsection{Axiomatisation of \texorpdfstring{$\FSklsum$}{FStoch\^oplus\_kl}}\label{sec:axiomfstoch}
In this section, we provide a quantitative monoidal axiomatisation of $\FSklsum$. We closely follow the structure of \cref{sec:axiombstoch}, relying now on the monoidal presentation of $\FSsum$ in~\cite{Fritz09}.

\begin{definition}\label{defn:thklsum}
	Let $\sig$ be the signature of $\th_{\!\ConvAlg}$. The theory $\th_{\!\KLsum}$ is defined with the same signature $\sig$ and the following $\zeroinfQ$-implications as axioms.
	\begin{itemize}[nosep]
		\item For any axiom $f=g$ in $\th_{\!\ConvAlg}$, $\th_{\!\KLsum}$ contains $\ \Rightarrow f=g$.
		\item For any $f: n \rightarrow m \in \vync{\sig}$, $\th_{\!\KLsum}$ contains $\ \Rightarrow  f=_0 f$.
		\item For any $f,g: 1 \rightarrow n$, $f',g': 1 \rightarrow m$, and $p,q \in [0,1]$, $\th_{\!\KLsum}$ contains $\Chainsum$.
		\item For any $f,g: n \rightarrow m$ and $f',g': n' \rightarrow m'$, $\th_{\!\KLsum}$ contains $\Parmax$.
		      \[\begin{bprooftree}
				      \AxiomC{$\scalebox{0.7}{$\distdiag{f}$} =_\varepsilon \scalebox{0.7}{$\distdiag{g}$}$}
				      \AxiomC{$\scalebox{0.7}{$\distdiag{f'}$} =_{\delta} \scalebox{0.7}{$\distdiag{g'}$}$}
				      \RightLabel{\textnormal{$\Chainsum$}}
				      \BinaryInfC{$\scalebox{0.7}{$\input{pffprime.tikz}$} =_{\kl(\raisebox{0.3ex}{$\distvect{p}{1-p}$},\raisebox{0.3ex}{$\distvect{q}{1-q}$})+p\varepsilon + (1-p)\delta} \scalebox{0.7}{$\input{qggprime.tikz}$}$}
			      \end{bprooftree}\qquad \begin{bprooftree}
				      \AxiomC{$\scalebox{0.7}{$\begin{tikzpicture}
	\begin{pgfonlayer}{nodelayer}
		\node [style=basic box] (24) at (0, 0) {$f$};
		\node [style=none] (25) at (1, 0) {};
		\node [style=none] (27) at (-1, 0) {};
	\end{pgfonlayer}
	\begin{pgfonlayer}{edgelayer}
		\draw (24) to (25.center);
		\draw (27.center) to (24);
	\end{pgfonlayer}
\end{tikzpicture}
$} =_{\varepsilon} \scalebox{0.7}{$\begin{tikzpicture}
	\begin{pgfonlayer}{nodelayer}
		\node [style=basic box] (24) at (0, 0) {$g$};
		\node [style=none] (25) at (1, 0) {};
		\node [style=none] (27) at (-1, 0) {};
	\end{pgfonlayer}
	\begin{pgfonlayer}{edgelayer}
		\draw (24) to (25.center);
		\draw (27.center) to (24);
	\end{pgfonlayer}
\end{tikzpicture}
$}$}
				      \AxiomC{$\scalebox{0.7}{$\begin{tikzpicture}
	\begin{pgfonlayer}{nodelayer}
		\node [style=basic box] (24) at (0, 0) {$f'$};
		\node [style=none] (25) at (1, 0) {};
		\node [style=none] (27) at (-1, 0) {};
	\end{pgfonlayer}
	\begin{pgfonlayer}{edgelayer}
		\draw (24) to (25.center);
		\draw (27.center) to (24);
	\end{pgfonlayer}
\end{tikzpicture}
$} =_{\delta} \scalebox{0.7}{$\begin{tikzpicture}
	\begin{pgfonlayer}{nodelayer}
		\node [style=basic box] (24) at (0, 0) {$g'$};
		\node [style=none] (25) at (1, 0) {};
		\node [style=none] (27) at (-1, 0) {};
	\end{pgfonlayer}
	\begin{pgfonlayer}{edgelayer}
		\draw (24) to (25.center);
		\draw (27.center) to (24);
	\end{pgfonlayer}
\end{tikzpicture}
$}$}
				      \RightLabel{\textnormal{$\Parmax$}}
				      \BinaryInfC{$\scalebox{0.7}{$\input{parletters.tikz}$} =_{\max\{ \varepsilon,\delta \}} \scalebox{0.7}{$\input{parletterstwo.tikz}$}$}
			      \end{bprooftree}\]
	\end{itemize}
	We write $\sync{\KLsum}$ for the syntactic $\zeroinfQ$-SMC generated by $\th_{\!\KLsum}$ (see \cref{defn:synvcat}).
\end{definition}
Our axiomatisation result amounts to showing that the functor $\sem{-}$ defined in \cref{eqn:semfssum} yields an isomorphism of enriched categories between $\sync{\KLsum}$ and $\FSklsum$. First, we show $\sem{-}$ satisfies the axioms in $\th_{\!\KLsum}$.

\begin{lemma}\label{lem:semsumsat}
	The $\zeroinfQ$-implications in $\th_{\!\KLsum}$ are satisfied by $\sem{-}: \vync{\sig} \rightarrow \FSklsum$. 
\end{lemma}
\begin{proof}
	As in the proof of \cref{lem:semprodsat}, $\sem{-}$ factors through an isomorphism $\sync{\ConvAlg} \cong \FSsum$ by~\cite[Theorem~3.14]{Fritz09}, and it follows that $\sem{f} = \sem{g}$ whenever $f=g$ is in $\th_{\!\ConvAlg}$. The implications $\ \Rightarrow f=_0 f$ are satisfied because $\kl(\dist,\dist)=0$ for any $\dist$. The implications $\Chainsum$ and $\Parmax$ are satisfied by \cref{eqn:chainsumdiags} and by \cref{eqn:klnexppar} respectively (we also use \cref{eqn:monotoneklsum}).
\end{proof}

By \cref{lem:syntactic}, we obtain an enriched symmetric monoidal functor $\sem{-}: \sync{\KLsum} \rightarrow \FSklsum$. Again, this is simply the isomorphism in the axiomatisation result of~\cite{Fritz09}, which is enriched. In particular, the underlying SMCs of $\sync{\KLsum}$ and $\FSklsum$ are $\sync{\ConvAlg}$ and $\FSsum$ respectively. 

It remains to show that $\sem{-}$ is locally an isometry. We first tackle the case of terms $1 \rightarrow m$.
\begin{lemma}\label{lem:isometryondist}
	For any $f,g: 1 \rightarrow m \in \sync{\KLsum}$, $\ \Rightarrow  f=_{\kl(\sem{f},\sem{g})} g$ is derivable in $\th_{\!\KLsum}$.
\end{lemma}
\begin{proof}
	We proceed by induction on $m$. For $m = 0$ or $m=1$, the result is trivial because there is no morphism in $\sync{\KLsum}(1,0)$ and a single one in $\sync{\KLsum}(1,1)$.

	Let $f, g: 1 \rightarrow m$ be two morphisms in $\sync{\KLsum}$. By reasoning solely within $\th_{\!\ConvAlg}$, we can decompose $f$ and $g$ as follows, where $p$ is first entry of $\sem{f}$, $q$ is that for $\sem{g}$, and $f'$ and $g'$ correspond the columns consisting of the remaining entries after normalisation. (This is essentially \cite[Lemma~3.2]{Fritz09} written recursively.) 
	\[\scalebox{0.7}{$\distdiag{f}$} = \scalebox{0.7}{$\input{decomposedist.tikz}$}  \qquad \scalebox{0.7}{$\distdiag{g}$} = \scalebox{0.7}{$\input{decomposedistb.tikz}$}\]
	Our induction hypothesis says that there is a proof of $f'=_{\varepsilon'} g'$ where $\varepsilon' = \kl(\sem{f'},\sem{g'})$. By applying the $\Chainsum$ rule (and $\id_1 =_0 \id_1$), we get $f=_{\varepsilon} g$ with $\varepsilon = \kl(\raisebox{0.3ex}{$\distvect{p}{1-p}$},\raisebox{0.3ex}{$\distvect{q}{1-q}$})+ (1-p)\varepsilon'$. By \cref{eqn:klsumchainrule}, we find $\varepsilon = \kl(\dist,\distb)$ where $\dist = p\dirac{1}+(1-p)\sem{f'} = \sem{f}$ and $\distb=q\dirac{1}+(1-q)\sem{g'} = \sem{g}$. We thus derived $f=_{\kl(\sem{f},\sem{g})} g$.
\end{proof}
The rest of the proof follows the axiomatisation of total variation distance in~\cite{Lobbia2025}.

\begin{theorem}\label{thm:axfssumkl}
	The function $f \mapsto \sem{f}$ is a bijective isometry $\sync{\KLsum}(n,m) \rightarrow \FSklsum(n,m)$. 
\end{theorem}
\begin{proof}
	We already know it is a nonexpansive function because $\sem{-}$ is an enriched functor. It is a bijection because $\sem{-}$ is an isomorphism of the underlying SMCs, $\sync{\ConvAlg}$ and $\FSsum$. Thus, we only need to show that $\klmax(\sem{f},\sem{g}) \geq d^{\KLsum}(f,g)$ for any $f,g: n \rightarrow m$. By \cite[Propositions~3.12 and 3.13]{Fritz09}, we have a decomposition $f = (f_1 \otimes \cdots \otimes f_n)\scomp p_m^n$ and $g = (g_1 \otimes \cdots \otimes g_n)\scomp p_m^n$,
	where $f_i: 1 \rightarrow m$ and $\sem{f_i}$ is the $i$th column of $\sem{f}$, similarly for $g$, and $p_m^n : nm \rightarrow m$. By \cref{lem:isometryondist}, $f_i =_{\kl(\sem{f_i},\sem{g_i})} g_i$ is derivable in $\th_{\!\KLsum}$ for each $i$, and $p_m^n =_0 p_m^n$ is derived by the second axiom in \cref{defn:thklsum}. Then, we can apply \cref{eqn:implicationnexpseq} and $\Parmax$, to derive $f=_{\varepsilon} g$ with $		\varepsilon = \max_{i \in \fset{n}}\kl(\sem{f_i},\sem{g_i}) + 0 \stackrel{\cref{eqn:klmax}}{=} \klmax(\sem{f},\sem{g})$.
	We conclude that the distance between $\sem{f}$ and $\sem{g}$ in $\sync{\KLsum}$ is below $\klmax(\sem{f},\sem{g})$ as desired.
\end{proof}

\section{Axiomatisation of R\'{e}nyi Divergences}\label{sec:axren}
KL divergence is an instance of a family of generalised relative entropies called R\'{e}nyi divergences~\cite{Renyi1960}, which are also $\zeroinfQ$-relations on $\Dset X$. These satisfy chain rules similar to \cref{lem:klprodchainrule,lem:klsumchainrule}, which means that our axiomatisations in \cref{sec:axkl} can be generalised as well. In fact, the proofs barely change.\footnote{Since KL divergence is an instance of the Rényi divergences, \cref{sec:axren} could have been written as a strict generalisation of \cref{sec:axkl}, but it would require more burdensome case analysis.}

In \cref{sec:rendiv}, we recall the definition of R\'{e}nyi divergences between distributions, and their extensions to stochastic matrices yielding two $\zeroinfQ$-SMCs $\BSreprod{\alpha}$ and $\FSresum{\alpha}$. In \cref{sec:axiomsrendiv}, we explain what changes to \cref{sec:axkl} lead to axiomatisations of these $\zeroinfQ$-SMCs.

\subsection{R\'{e}nyi Divergences}\label{sec:rendiv}
\begin{definition}
	Given a parameter $0 < \alpha < \infty$ with $\alpha \neq 1$, the \emph{R\'{e}nyi divergence of order $\alpha$} is a map $\Dset X \times \Dset X \rightarrow \zeroinfQ$ defined by $\re{\alpha}(\dist,\distb) \coloneq \frac{1}{\alpha-1}\log(\sum_{x \in X}\frac{\dist(x)^\alpha}{\distb(x)^{\alpha-1}})$, where $\frac{x}{0}$ is $0$ when $x=0$ and $\infty$ otherwise, $\log(\infty) = \infty$, and $0^0 = 0$. For $\alpha = 0,1,\infty$, the R\'{e}nyi divergence is defined as a limit, or equivalently with the following formulas. 
	\[\re{0}(\dist,\distb) = -\log(\sum_{x \in X, \dist(x) > 0}\!\!\!\!\!\!\!\!\distb(x)) \qquad \re{1}(\dist,\distb) = \kl(\dist,\distb) \qquad \re{\infty}(\dist,\distb) = \sup_{x \in X}\log\frac{\dist(x)}{\distb(x)}\]
\end{definition}

The R\'{e}nyi divergences also satisfy a chain rule (see e.g.~\cite[Section 7.12]{Polyanskiy2025}). We state the specific variants that we will need (\emph{cf.}~\cref{lem:klprodchainrule,lem:klsumchainrule}), omitting the case $\alpha=1$ which is the KL divergence.
\begin{lemma}\label{lem:reprodchainrule}
	Given $0\leq \alpha<\infty$ with $\alpha\neq 1$ and $\dist,\distb\in\Dset(\B\times Y)$ with $p \coloneq \dist(\b1,Y)$ and $q \coloneq \distb(\b1,Y)$,
	\begin{equation}\label{eqn:reprodrecurse}
		\re{\alpha}(\dist,\distb)
		=
		\frac{1}{\alpha-1}\log\big(
		\frac{p^\alpha}{q^{\alpha-1}}e^{(\alpha-1)\re{\alpha}(\overline{\dist}(\b1,-),\overline{\distb}(\b1,-))}
		+
		\frac{(1-p)^\alpha}{(1-q)^{\alpha-1}}e^{(\alpha-1)\re{\alpha}(\overline{\dist}(\b0,-),\overline{\distb}(\b0,-))}
		\big).
	\end{equation}
	For $\alpha = \infty$, we have
	\begin{equation}\label{eqn:reprodrecurseinfty}
		\re{\infty}(\dist,\distb)
		= \max \big\{ \log \frac{p}{q} + \re{\infty}(\overline{\dist}(\b1,-),\overline{\distb}(\b1,-)), \log\frac{1-p}{1-q} + \re{\infty}(\overline{\dist}(\b0,-),\overline{\distb}(\b0,-))\big\}.
	\end{equation}
\end{lemma}
\begin{lemma}\label{lem:rerecursivedefn}
	Given $0\leq\alpha<\infty$ with $\alpha\neq 1$, $\dist,\distb\in \Dset(X+Y)$, $p \coloneq \sum_{x \in X}\dist(x)$ and $q \coloneq \sum_{x \in X}\distb(x)$, writing $\dist_X$ for the conditional distribution on $X$, and similarly for $\dist_Y$, $\distb_X$, and $\distb_Y$, we have,
	\begin{equation}\label{eqn:rerecursivedefn}
		\re{\alpha}(\dist,\distb) = \frac{1}{\alpha-1}\log\big(\frac{p^\alpha}{q^{\alpha-1}}e^{(\alpha-1)\re{\alpha}(\dist_X,\distb_X)}+\frac{(1-p)^\alpha}{(1-q)^{\alpha-1}}e^{(\alpha-1)\re{\alpha}(\dist_Y,\distb_Y)}
		\big).
	\end{equation}
	For $\alpha = \infty$, we have
	\begin{equation}\label{eqn:rerecursivedefninfty}
		\re{\infty}(\dist,\distb)
		= \max \big\{ \log \frac{p}{q} + \re{\infty}(\dist_X,\distb_X), \log\frac{1-p}{1-q} + \re{\infty}(\dist_Y,\distb_Y)\big\}.
	\end{equation}
\end{lemma}
In order to avoid writing the long formulas in \cref{eqn:reprodrecurse,eqn:reprodrecurseinfty,eqn:rerecursivedefn,eqn:rerecursivedefninfty} many times, we will use the following shorthand function $C_{\alpha}: [0,1]^2 \times [0,\infty]^2 \rightarrow [0,\infty]$.
\begin{equation}\label{eqn:shorthandcalpha}
	C_{\alpha}(p,q,\varepsilon,\delta) = \begin{cases}
		\frac{1}{\alpha-1}\log( \frac{p^{\alpha}}{q^{\alpha-1}}e^{(\alpha-1)\varepsilon} + \frac{(1-p)^\alpha}{(1-q)^{\alpha-1}}e^{(\alpha-1)\delta} ) & 0\leq \alpha < \infty, \alpha \neq 1 \\
		\max \{ \log \frac{p}{q} + \varepsilon, \log\frac{1-p}{1-q} + \delta\}                                                                         & \alpha = \infty
	\end{cases}
\end{equation}
This function describes how the divergence between two joint distributions depends on the divergences between their conditionals. Note that for any $\alpha$, $C_{\alpha}$ is monotone in the third and fourth arguments. Namely, we have analogues to \cref{eqn:monotoneklprod,eqn:monotoneklsum}, which means we can transform the chain rules of \cref{lem:reprodchainrule,lem:rerecursivedefn} into sound $\zeroinfQ$-implications in the next section.

\subsection{Axiomatisations}\label{sec:axiomsrendiv}
For the rest of this section, we fix the parameter $\alpha \in [0,\infty]$ with $\alpha \neq 1$. We extend the definition of R\'{e}nyi divergence of order $\alpha$ to stochastic matrices following \cref{eqn:klmax}: given two $m\times n$ stochastic matrices, let $\remax{\alpha}(A,B) \coloneq \max_{i \in \fset{n}} \re{\alpha}(A_i,B_i)$.
Perrone showed in~\cite[Proposition~2.18]{Perrone2024} that this definition satisfies \cref{eqn:nexpcomposition,eqn:nexptensor}, making $\FSprod$ equipped with $\remax{\alpha}$ on its homsets into a $\zeroinfQ$-SMC. This in turn yields a $\zeroinfQ$-SMC $\BSreprod{\alpha}$, that is axiomatised by the following $\zeroinfQ$-theory.
\begin{definition}\label{defn:threprod}
	Let $\sig$ be the signature of $\th_{\!\Circ}$. The theory $\th_{\!\Rprod{\alpha}}$ is defined with the same signature $\sig$ and the following $\zeroinfQ$-implications as axioms.
	\begin{itemize}[nosep]
		\item For any axiom $f=g$ in $\th_{\!\Circ}$, $\th_{\!\Rprod{\alpha}}$ contains $\ \Rightarrow f=g$.
		\item For any $f: n \rightarrow m \in \vync{\sig}$, $\th_{\!\Rprod{\alpha}}$ contains $\ \Rightarrow  f=_0 f$.
		\item For any $f_{\b1},g_{\b1}: 0 \rightarrow m$, $f_{\b0},g_{\b0}: 0 \rightarrow m$, and $p,q \in [0,1]$, $\th_{\!\Rprod{\alpha}}$ contains $\Chainprod$.
		\item For any $f_{\b1},g_{\b1}: n \rightarrow m$ and $f_{\b0},g_{\b0}: n' \rightarrow m$, $\th_{\!\Rprod{\alpha}}$ contains $\Ifmax$.
		      \[\begin{bprooftree}
				      \AxiomC{$\raisebox{0.3ex}{\scalebox{0.7}{$\begin{tikzpicture}
	\begin{pgfonlayer}{nodelayer}
		\node [style=basic box] (24) at (0, 0) {$f_{\b1}$};
		\node [style=none] (25) at (1, 0) {};
	\end{pgfonlayer}
	\begin{pgfonlayer}{edgelayer}
		\draw [style=double] (24) to (25.center);
	\end{pgfonlayer}
\end{tikzpicture}
$}} =_{\varepsilon} \raisebox{0.3ex}{\scalebox{0.7}{$\begin{tikzpicture}
	\begin{pgfonlayer}{nodelayer}
		\node [style=basic box] (24) at (0, 0) {$g_{\b1}$};
		\node [style=none] (25) at (1, 0) {};
	\end{pgfonlayer}
	\begin{pgfonlayer}{edgelayer}
		\draw [style=double] (24) to (25.center);
	\end{pgfonlayer}
\end{tikzpicture}
$}}$}
				      \AxiomC{$\raisebox{0.3ex}{\scalebox{0.7}{$\begin{tikzpicture}
	\begin{pgfonlayer}{nodelayer}
		\node [style=basic box] (24) at (0, 0) {$f_{\b0}$};
		\node [style=none] (25) at (1, 0) {};
	\end{pgfonlayer}
	\begin{pgfonlayer}{edgelayer}
		\draw [style=double] (24) to (25.center);
	\end{pgfonlayer}
\end{tikzpicture}
$}} =_{\delta} \raisebox{0.3ex}{\scalebox{0.7}{$\begin{tikzpicture}
	\begin{pgfonlayer}{nodelayer}
		\node [style=basic box] (24) at (0, 0) {$g_{\b0}$};
		\node [style=none] (25) at (1, 0) {};
	\end{pgfonlayer}
	\begin{pgfonlayer}{edgelayer}
		\draw [style=double] (24) to (25.center);
	\end{pgfonlayer}
\end{tikzpicture}
$}}$}
				      \RightLabel{\textnormal{$\Chainprod$}}
				      \BinaryInfC{$\scalebox{0.7}{$\input{distdecompose.tikz}$} =_{C_{\alpha}(p,q,\varepsilon,\delta)} \scalebox{0.7}{$\input{distsbdecompose.tikz}$}$}
			      \end{bprooftree} \quad \begin{bprooftree}
				      \AxiomC{$\raisebox{0.3ex}{\scalebox{0.7}{$\begin{tikzpicture}
	\begin{pgfonlayer}{nodelayer}
		\node [style=basic box] (24) at (0, 0) {$f_{\b1}$};
		\node [style=none] (25) at (1, 0) {};
		\node [style=none] (27) at (-1, 0) {};
	\end{pgfonlayer}
	\begin{pgfonlayer}{edgelayer}
		\draw[style=double] (24) to (25.center);
		\draw[style=double] (27.center) to (24);
	\end{pgfonlayer}
\end{tikzpicture}
$}} =_{\varepsilon} \raisebox{0.3ex}{\scalebox{0.7}{$\begin{tikzpicture}
	\begin{pgfonlayer}{nodelayer}
		\node [style=basic box] (24) at (0, 0) {$g_{\b1}$};
		\node [style=none] (25) at (1, 0) {};
		\node [style=none] (27) at (-1, 0) {};
	\end{pgfonlayer}
	\begin{pgfonlayer}{edgelayer}
		\draw [style=double] (24) to (25.center);
		\draw [style=double] (27.center) to (24);
	\end{pgfonlayer}
\end{tikzpicture}
$}}$}
				      \AxiomC{$\raisebox{0.3ex}{\scalebox{0.7}{$\begin{tikzpicture}
	\begin{pgfonlayer}{nodelayer}
		\node [style=basic box] (24) at (0, 0) {$f_{\b0}$};
		\node [style=none] (25) at (1, 0) {};
		\node [style=none] (27) at (-1, 0) {};
	\end{pgfonlayer}
	\begin{pgfonlayer}{edgelayer}
		\draw [style=double] (24) to (25.center);
		\draw [style=double] (27.center) to (24);
	\end{pgfonlayer}
\end{tikzpicture}
$}} =_{\delta} \raisebox{0.3ex}{\scalebox{0.7}{$\begin{tikzpicture}
	\begin{pgfonlayer}{nodelayer}
		\node [style=basic box] (24) at (0, 0) {$g_{\b0}$};
		\node [style=none] (25) at (1, 0) {};
		\node [style=none] (27) at (-1, 0) {};
	\end{pgfonlayer}
	\begin{pgfonlayer}{edgelayer}
		\draw [style=double] (24) to (25.center);
		\draw [style=double] (27.center) to (24);
	\end{pgfonlayer}
\end{tikzpicture}
$}}$}
				      \RightLabel{{\textnormal{$\Ifmax$}}}
				      \BinaryInfC{$\scalebox{0.7}{$\input{iff.tikz}$} =_{\max\left\{ \varepsilon,\delta \right\}} \scalebox{0.7}{$\input{ifg.tikz}$}$}
			      \end{bprooftree}\]
	\end{itemize}
	We write $\sync{\Rprod{\alpha}}$ for the syntactic $\zeroinfQ$-SMC generated by $\th_{\!\Rprod{\alpha}}$ (see \cref{defn:synvcat}).
\end{definition}
\begin{theorem}\label{thm:axiomreprod}
	The functor $\sem{-}$ defined in \cref{eqn:sembsprod} is an enriched isomorphism of SMCs $\sync{\Rprod{\alpha}} \cong \BSreprod{\alpha}$.
\end{theorem}
\begin{proof}
	The proof follows that of \cref{sec:axiombstoch}. In particular,
	\begin{itemize}[nosep]
		\item soundness of the new $\Chainprod$ rule is justified by \cref{eqn:reprodrecurse} (or \cref{eqn:reprodrecurseinfty} when $\alpha = \infty$) and monotonicity of $C_{\alpha}$,
		\item soundness of $\Ifmax$ is justified by the same proof as for $\klmax$, and
		\item the use of \cref{eqn:klprodrecurse} in \cref{lem:prod:isometryondist} is replaced by \cref{eqn:reprodrecurse} (or \cref{eqn:reprodrecurseinfty} when $\alpha = \infty$).
	\end{itemize}
	\vspace{-1\baselineskip}
\end{proof}

We can argue as for $\FSklsum$ that $\FSsum$ equipped with $\remax{\alpha}$ on its homset is a $\zeroinfQ$-SMC. We denote it by $\FSresum{\alpha}$ and show it is axiomatised by the following $\zeroinfQ$-theory.
\begin{definition}\label{defn:thresum}
	Let $\sig$ be the signature of $\th_{\!\ConvAlg}$. The theory $\th_{\!\Rsum{\alpha}}$ is defined with the same signature $\sig$ and the following $\zeroinfQ$-implications as axioms.
	\begin{itemize}[nosep]
		\item For any axiom $f=g$ in $\th_{\!\ConvAlg}$, $\th_{\!\Rsum{\alpha}}$ contains $\ \Rightarrow f=g$.
		\item For any $f: n \rightarrow m \in \vync{\sig}$, $\th_{\!\Rsum{\alpha}}$ contains $\ \Rightarrow  f=_0 f$.
		\item For any $f,g: 1 \rightarrow n$, $f',g': 1 \rightarrow m$, and $p,q \in [0,1]$, $\th_{\!\Rsum{\alpha}}$ contains $\Chainsum$.
		\item For any $f,g: n \rightarrow m$ and $f',g': n' \rightarrow m'$, $\th_{\!\Rsum{\alpha}}$ contains $\Parmax$.
		      \[\begin{bprooftree}
				      \AxiomC{$\scalebox{0.7}{$\distdiag{f}$} =_\varepsilon \scalebox{0.7}{$\distdiag{g}$}$}
				      \AxiomC{$\scalebox{0.7}{$\distdiag{f'}$} =_{\delta} \scalebox{0.7}{$\distdiag{g'}$}$}
				      \RightLabel{\textnormal{$\Chainsum$}}
				      \BinaryInfC{$\scalebox{0.7}{$\input{pffprime.tikz}$} =_{C_{\alpha}(p,q,\varepsilon,\delta)} \scalebox{0.7}{$\input{qggprime.tikz}$}$}
			      \end{bprooftree}\qquad \begin{bprooftree}
				      \AxiomC{$\scalebox{0.7}{$\begin{tikzpicture}
	\begin{pgfonlayer}{nodelayer}
		\node [style=basic box] (24) at (0, 0) {$f$};
		\node [style=none] (25) at (1, 0) {};
		\node [style=none] (27) at (-1, 0) {};
	\end{pgfonlayer}
	\begin{pgfonlayer}{edgelayer}
		\draw (24) to (25.center);
		\draw (27.center) to (24);
	\end{pgfonlayer}
\end{tikzpicture}
$} =_{\varepsilon} \scalebox{0.7}{$\begin{tikzpicture}
	\begin{pgfonlayer}{nodelayer}
		\node [style=basic box] (24) at (0, 0) {$g$};
		\node [style=none] (25) at (1, 0) {};
		\node [style=none] (27) at (-1, 0) {};
	\end{pgfonlayer}
	\begin{pgfonlayer}{edgelayer}
		\draw (24) to (25.center);
		\draw (27.center) to (24);
	\end{pgfonlayer}
\end{tikzpicture}
$}$}
				      \AxiomC{$\scalebox{0.7}{$\begin{tikzpicture}
	\begin{pgfonlayer}{nodelayer}
		\node [style=basic box] (24) at (0, 0) {$f'$};
		\node [style=none] (25) at (1, 0) {};
		\node [style=none] (27) at (-1, 0) {};
	\end{pgfonlayer}
	\begin{pgfonlayer}{edgelayer}
		\draw (24) to (25.center);
		\draw (27.center) to (24);
	\end{pgfonlayer}
\end{tikzpicture}
$} =_{\delta} \scalebox{0.7}{$\begin{tikzpicture}
	\begin{pgfonlayer}{nodelayer}
		\node [style=basic box] (24) at (0, 0) {$g'$};
		\node [style=none] (25) at (1, 0) {};
		\node [style=none] (27) at (-1, 0) {};
	\end{pgfonlayer}
	\begin{pgfonlayer}{edgelayer}
		\draw (24) to (25.center);
		\draw (27.center) to (24);
	\end{pgfonlayer}
\end{tikzpicture}
$}$}
				      \RightLabel{\textnormal{$\Parmax$}}
				      \BinaryInfC{$\scalebox{0.7}{$\input{parletters.tikz}$} =_{\max\{ \varepsilon,\delta \}} \scalebox{0.7}{$\input{parletterstwo.tikz}$}$}
			      \end{bprooftree}\]
	\end{itemize}
	We write $\sync{\Rsum{\alpha}}$ for the syntactic $\zeroinfQ$-SMC generated by $\th_{\!\Rsum{\alpha}}$ (see \cref{defn:synvcat}).
\end{definition}
\begin{theorem}\label{thm:axiomresum}
	The functor $\sem{-}$ defined in \cref{eqn:semfssum} is an enriched isomorphism of SMCs $\sync{\Rsum{\alpha}} \cong \FSresum{\alpha}$.
\end{theorem}
\begin{proof}
	The proof follows that of \cref{sec:axiomfstoch}. In particular,
	\begin{itemize}[nosep]
		\item soundness of the new $\Chainsum$ rule is justified by \cref{eqn:rerecursivedefn} (or \cref{eqn:rerecursivedefninfty} when $\alpha = \infty$) and monotonicity of $C_{\alpha}$,
		\item soundness of $\Parmax$ is justified by the same proof as for $\klmax$,
		\item and the use of \cref{eqn:klsumchainrule} in \cref{lem:isometryondist} is replaced by \cref{eqn:rerecursivedefn} (or \cref{eqn:rerecursivedefninfty} when $\alpha = \infty$).
	\end{itemize}
	\vspace{-1\baselineskip}
\end{proof}
\section{Conclusions}\label{sec:conclusions}

In this paper, we have provided axiomatisations of relative entropy as an enrichment of $\FStoch$, leveraging string diagrammatic reasoning and quantitative equations. This represents a new application of quantitative equational reasoning~(\!\!\cite{Mardare2016}) to distances between probabilistic processes, addressing a notable gap in the literature where relative entropy had not previously been treated within this framework.

While our development takes place at the level of monoidal categories, the theories $\th_{\KLsum}$ and $\th_{\Rsum{\alpha}}$ are based on cocartesian theories~(\!\!\cite{Bonchi2018}). Consequently, by considering them as presentations of the opposite (cartesian) category, they can be translated into quantitative algebraic theories following the correspondence described in~\cite[Section~6]{Lobbia2025}. The resulting presentation is a theory of convex algebras equipped with two additional quantitative inferences, expressed in the style of~\cite{Mardare2016} (without the metric axioms):
\[
	\vdash x =_0 x \qquad \text{and} \qquad x =_{\varepsilon} x', y =_{\delta} y' \vdash x +_p y =_{C(p,q,\varepsilon,\delta)} x' +_q y',
\]
where $C$ is the function $C_{\alpha}$ defined in \cref{eqn:shorthandcalpha} or $C_1(p,q,\varepsilon,\delta) = \kl(\raisebox{0.3ex}{$\distvect{p}{1-p}$}, \raisebox{0.3ex}{$\distvect{q}{1-q}$}) + p\varepsilon + (1-p)\delta$. This may serve as a quantitative algebraic axiomatisation of relative entropy between probability distributions, rather than generic stochastic matrices as considered in the present paper.
The monads induced by these theories map a discrete space on a set $X$ to $\Dset X$ equipped with the corresponding relative entropy. Extending this description to non-discrete spaces remains a direction for future work.

By contrast, the theories $\th_{\KLprod}$ and $\th_{\Rprod{\alpha}}$ are neither cartesian nor cocartesian, and therefore require the abstraction of the monoidal setting. Their axiomatisations contribute to recent developments in categorical probability, where quantitative reasoning via enrichment is emerging, e.g.\@ in~\cite{Perrone2024}. In that work, divergence is treated abstractly and instantiated to specific relative entropies, yielding only a general formulation of the chain rule. Our rules $\Chainprod$ and $\Chainsum$, in contrast, are sufficiently precise to characterise the corresponding relative entropies uniquely. This  opens the possibility of synthetic proofs for results that depend on the specific properties of KL or R\'enyi divergences.

Another promising direction concerns the string diagrammatic study of Bayesian reasoning~(\!\!\cite{Jacobs2021,JacobsKZ21,Jacobs2023,Lorenzin2025}). KL divergence plays a central role in evaluating probabilistic learning models~(\!\!\cite{Fox2012,KingmaW13}); our framework suggests the possibility of proving mathematical properties of model performances using  quantitative diagrammatic reasoning.

The generalisation from quantitative equations to quantitative implications was essential to accommodate the enrichments introduced in \cref{sec:axkl,sec:axren} using the chain rules. Nevertheless, the theoretical foundations of $\vV$-theories merit further clarification. While models may be regarded as functors out of the syntactic category~(\cref{lem:syntactic}), the converse does not generally hold, resulting in an incomplete functorial semantics. Adapting the functorial semantics of implicational theories~(\!\!\cite{Barr1989,Adamek1998}) to the setting of monoidal algebra remains a question to be explored in future work. We also conjecture that the comparison established in~\cite[Section~6]{Lobbia2025} extends to $\vV$-theories: cartesian $\vV$-theories (equipped with natural copy and discard structure) should correspond to quantitative algebraic theories as developed in~\cite{Mardare2016}.
Our treatment of implicational axioms also allows to accommodate the different choices of inference rules in~\cite{Lobbia2025} more directly. Namely, the choice can be baked in the theory (as axioms) instead of in the generation of the free model as was the case in~\cite{Lobbia2025}. This is strictly more general than \cite{Lobbia2025}, and it encompasses the example in~\cite{diGiorgio2025}, which provides an axiomatisation of the total variation distance on $\BSprod$ that does not fit within the framework of quantitative monoidal algebra developed in~\cite{Lobbia2025}. 

Finally, relative entropies have quantum counterparts that are relevant in quantum information and computing~(\!\!\cite{Vedral2002}). Given the widespread adoption of diagrammatic reasoning in quantum settings~(\!\!\cite{DodoPQP}), extending our results to diagrammatic theories of quantum processes is a natural and promising direction. Notably, quantum relative entropy has already been studied categorically, using string diagrams, in the context of quantum natural language processing~(\!\!\cite{Balkir2016}).
\bibliographystyle{./entics}
\bibliography{refs}

\end{document}